\pdfoutput=1

\documentclass[english,notitlepage]{amsart}
\usepackage[T1]{fontenc}
\usepackage[latin9]{inputenc}
\usepackage{babel}
\usepackage{mathtools}
\usepackage{amsmath}
\usepackage{amsfonts}
\usepackage{amsthm}
\usepackage{amssymb}
\usepackage{hyperref}
\usepackage{xspace}
\usepackage{booktabs}

\usepackage{graphics}
\usepackage{epstopdf}

\usepackage[capitalize]{cleveref}
\usepackage{fullpage}
\usepackage{algorithm}
\usepackage{algpseudocode}

\usepackage{colortbl}

\usepackage{lmodern}
\usepackage{microtype}
\usepackage{qcircuit}

\usepackage{amsaddr}

\usepackage{pdflscape}
\usepackage{afterpage}

\usepackage{tikz}
\usetikzlibrary{matrix}


\theoremstyle{plain}
\newtheorem{theorem}{Theorem}[section]
\newtheorem{lemma}[theorem]{Lemma}
\newtheorem{proposition}[theorem]{Proposition}

\newtheorem{corollary}[theorem]{Corollary}

\theoremstyle{definition}
\newtheorem{definition}[theorem]{Definition}
\newtheorem{example}[theorem]{Example}

\theoremstyle{remark}
\newtheorem{remark}[theorem]{Remark}
\newtheorem*{notation*}{Notation}

\newcommand{\F}{\mathbb{F}_2}
\newcommand{\Z}{\mathbb{Z}}
\newcommand{\R}{\mathbb{R}}
\newcommand{\N}{\mathbb{N}}

\newcommand{\GL}[1]{\mathrm{GL}(#1, \F)}

\newcommand{\argmax}{\operatorname*{arg\,max}}

\newcommand{\supp}{\operatorname{supp}}

\newcommand{\concat}{\mathbin{\|}}

\newcommand{\iso}{\simeq}

\newcommand{\ket}[1]{|#1\rangle}
\newcommand{\bra}[1]{\langle#1|}

\newcommand{\cnot}{\mathrm{CNOT}}
\newcommand{\NOT}{\mathrm{NOT}}
\newcommand{\rz}{R_Z}

\renewcommand{\vec}[1]{\boldsymbol{#1}}

\newcommand{\w}{\vec{w}}
\newcommand{\x}{\vec{x}}
\newcommand{\y}{\vec{y}}
\newcommand{\z}{\vec{z}}
\newcommand{\0}{\vec{0}}

\newcommand{\append}{::}

\newcommand{\fourier}[1]{\widehat{#1}}
\newcommand{\parity}[1]{\chi_{#1}}

\newcommand{\etal}{et al.\xspace}
\usepackage{todonotes}

\begin{document}

\tikzset{
    node style bb/.style={rectangle,fill=black!20},
    node style ba/.style={red,rectangle,fill=black!20},
    node style aa/.style={red}
}

\title{On the CNOT-complexity of CNOT-PHASE circuits}

\author{Matthew Amy$^{1,2}$, Parsiad Azimzadeh$^{2}$ and Michele Mosca$^{1,3,4,5}$}
\address{$^1$ Institute for Quantum Computing, University of Waterloo, Canada \\
	$^2$ David R. Cheriton School of Computer Science, University of Waterloo, Canada \\
	$^3$ Department of Combinatorics \& Optimization, University of Waterloo, Canada \\
	$^4$ Perimeter Institute for Theoretical Physics, Waterloo, Canada \\
	$^5$ Canadian Institute for Advanced Research, Toronto, Canada}

\begin{abstract}
We study the problem of $\cnot$-optimal quantum circuit synthesis over gate sets consisting of $\cnot$ and $Z$-basis rotations of arbitrary angles. We show that the circuit-polynomial correspondence relates such circuits to Fourier expansions of pseudo-Boolean functions, and that for certain classes of functions this expansion uniquely determines the minimum $\cnot$ cost of an implementation. As a corollary we prove that $\cnot$ minimization over $\cnot$ and phase gates is at least as hard as synthesizing a $\cnot$-optimal circuit computing a set of parities of its inputs. We then show that this problem is NP-complete for two restricted cases where all $\cnot$ gates are required to have the same target, and where the circuit inputs are encoded in a larger state space. The latter case has applications to $\cnot$ optimization over more general Clifford+$T$ circuits.

We further present an efficient heuristic algorithm for synthesizing circuits over $\cnot$ and $Z$-basis rotations with small $\cnot$ cost. Our experiments show a 23\% reduction of $\cnot$ gates on average across a suite of Clifford+$T$ benchmark circuits, with a maximum reduction of 43\%.
\end{abstract}

\maketitle

%%%%%%%%%%%%%%%%%%%%%%%%%%%%%%%%%%%%%%%%
\section{Introduction}

The two-qubit controlled-$\NOT$ ($\cnot$) gate forms the backbone of most discrete quantum circuits. As the only entangling operation -- and, moreover, the only two-qubit operation -- in most common gate sets, it is used judiciously in effectively any practically useful quantum circuit. Even in physical implementations where gates commonly have tunable parameters many use the $\cnot$ gate, or a $\cnot$ gate up to single qubit rotations, as the two qubit entangling gate (e.g. \cite{slflwm16}). 

Due to the low cost of the controlled-$\NOT$ gate in most fault tolerant models \cite{fsg09, g98-2} compared to single-qubit gates implemented with resource states -- for instance, the $T$ gate in the surface code or the Hadamard gate in the 15-qubit Reed-Muller code -- reducing the number of $\cnot$ gates is typically regarded as a secondary optimization objective. In some cases, optimizations aimed at reducing $T$-count or $T$-depth (the number of layers of parallel $T$ gates in a circuit) result in a massive increase in $\cnot$ count \cite{amm14}. While in many contexts the savings obtained by optimizing $T$ gates or other such gates makes up for the explosion in $\cnot$ gates, it is nevertheless desirable to find optimization methods which can mitigate this increase, as even in fault tolerant models $\cnot$ gates incur non-negligible cost.

Direct optimization of $\cnot$ gates has generally been studied only in the restricted case of \emph{reversible} circuits. Synthesis and optimization methods have been developed for the subset of \emph{linear} reversible functions, which are exactly the circuits implementable with just $\cnot$ gates \cite{iky02, pmh08}, but the vast majority of optimizations apply to the NCT gate set, consisting over $\NOT$, $\cnot$ and Toffoli gates, or its generalization to multiply-controlled Toffoli gates. In the former case, additional reductions may be possible between distinct linear reversible circuits separated by other quantum gates. Likewise, further reductions are typically possible between Toffoli or multiply-controlled Toffoli gates once they have been expanded into suitable quantum gate sets \cite{aads16}.

In this work, we consider the problem of minimizing $\cnot$ count in the presence of other quantum gates, specifically phase gates (single qubit $Z$-basis rotations). By the \emph{circuit-polynomial correspondence} \cite{am16, dhmhno05,m17}, such circuits are known to correspond to weighted sums of parity functions called \emph{phase polynomials} \cite{amm14}. 

Using this correspondence, we introduce the problem of synthesizing a \emph{parity network}, a $\cnot$ circuit which computes a set of parities, possibly non-simultaneously. Computing a minimal parity network for a particular set of parity functions is shown to be equivalent to finding a $\cnot$-optimal circuit \emph{for a particular phase polynomial}, and to be computationally \emph{easier} than the full $\cnot$ optimization problem over $\cnot$ and phase gates. This problem is then shown to be NP-complete in two restricted cases: when all $\cnot$ gates are restricted to the same target bit, and when the $m$ circuit inputs are encoded in the state space of $n>m$ qubits. The former case provides evidence for the hardness of computing minimal parity networks, while the latter case is useful when optimizing $\cnot$ counts in Clifford+T circuits.

We further devise a new heuristic optimization algorithm for $\cnot$-phase circuits by synthesizing parity networks. The optimization algorithm is inspired by \emph{Gray codes} \cite{g53}, which cycle through the set of $n$-bit strings using the exact minimal number of bit flips. Like Gray codes, our algorithm achieves the minimal number of $\cnot$ gates when all $2^n$ parities are needed. To test our algorithm, we implemented it as a replacement for the $T$-parallelization sub-routine in the $T$-par framework \cite{amm14} to apply our optimization over general Clifford+$T$ circuits -- our experiments show a reduction in $\cnot$ count of 23\% on average.

The paper is organized as follows. Section~2 introduces the circuit-polynomial correspondence for quantum circuits composed of $\cnot$ and phase gates, as well as the minimal parity network synthesis problem, then further shows that the general problem of $\cnot$-minimization over $\cnot$ and phase gates is at least as hard as the parity network synthesis problem. Section~3 studies the complexity of this synthesis problem and shows that the fixed-target and encoded input cases are NP-complete. In Section~4 we present our heuristic $\cnot$ minimization algorithm, and Section~5 gives experimental results.

\subsection{Preliminaries}

We give a brief overview of the basics of quantum circuits and our notation throughout.

In this paper we work in the circuit model of quantum computation \cite{nc00}. The state of an $n$-qubit quantum system is taken as a unit vector in a dimension $2^n$ complex vector space. As is standard we denote the $2^n$ basis vectors of the \emph{computational} basis by $\ket{\x}$ for bit strings $\x\in\F^n$ -- these are called the \emph{classical} states. A general quantum state may then be written as a \emph{superposition} of classical states
\[
  \ket{\psi} = \sum_{\x\in\F^n} \alpha_{\x}\ket{\x},
\]
for complex $\alpha_{\x}$ and having unit norm. The states of two $n$ and $m$ qubit quantum systems $\ket{\psi}$ and $\ket{\phi}$ may be combined into an $n+m$ qubit state by taking their tensor product $\ket{\psi}\otimes\ket{\phi}$ -- in particular, $\ket{\x}\otimes \ket{\y} = \ket{\x \concat \y}$ where $\concat$ denotes the concatenation of two binary vectors or strings. If to the contrary the state of two qubits cannot be written as a tensor product the two qubits are said to be \emph{entangled}. We write $|\x|$ to denote the \emph{Hamming weight} of the binary string $\x$.

Quantum circuits, in analogy to classical circuits, carry qubits from left to right along \emph{wires} through \emph{gates} which transform the state. In the unitary circuit model gates are required to implement unitary operators on the state space -- that is, quantum gates are modelled by complex-valued matrices $U$ satisfying $UU^\dagger = U^\dagger U = I$, where $U^\dagger$ is the complex conjugate of $U$. As a result, such quantum computations must be \emph{reversible}, and in particular, the set of classical functions computable by a quantum circuit is the set of invertible, $n$ bit to $n$ bit functions. To perform general (irreversible) classical functions, extra bits called \emph{ancillae} are typically needed to implement them reversibly. We use $U_C$ to denote the unitary operator corresponding to the circuit $C$ and $C \append C'$ to denote the circuit obtained by appending $C'$ to the end of $C$.

We will primarily be interested in two quantum gates in this paper, the controlled-NOT gate ($\cnot$) which (as a function of classical states) inverts its second argument conditioned on the value of its first, and the $Z$-basis rotation $\rz(\theta)$ which applies a \emph{phase shift} of $e^{i\theta}$ conditioned on its argument. Specifically, we define
\[
  \cnot\ket{x}\ket{y} = \ket{x}\ket{x\oplus y}, \qquad \rz(\theta)\ket{x} = e^{2\pi i\theta x}\ket{x},
\]
for $x,y\in\F$. By the linearity of quantum computation a unitary operator is fully defined by its effect on the computational basis states. Other common quantum gates include the NOT gate $X\ket{x} = \ket{1\oplus x}$ and the Hadamard gate $H\ket{x} = \frac{1}{\sqrt{2}}\sum_{y\in\F}(-1)^{xy}\ket{y}$. The $\cnot$ gate together with the Hadamard gate and $S=\rz\left(\frac{1}{4}\right)$ generate the \emph{Clifford} group, a set of quantum operations particularly important in quantum error correction and for which there are efficient implementations in many fault-tolerant schemes. While the Clifford group alone is not \emph{universal} for quantum computing, in the sense that not every unitary matrix can be approximated to arbitrary accuracy by a Clifford group circuit, adding the $T=\rz\left(\frac{1}{8}\right)$ gate to the Clifford group gives a universal gate set. The Clifford group combined with the $T$ gate is typically referred to as the \emph{Clifford+$T$} gate set. We say a circuit is written over a particular set of gates if the circuit only uses gates from that set.

\subsection{Related work}

Previous work regarding $\cnot$ optimization has largely focussed on strictly reversible circuits. Iwama, Kambayashi and Yamashita \cite{iky02} gave some transformation rules which they use to normalize and optimize $\cnot$-based circuits. More specific to $\cnot$ circuits, Patel, Markov and Hayes \cite{pmh08} gave an algorithm for synthesizing linear reversible circuits which produces circuits of asymptotically optimal size. Their method modifies Gaussian elimination by prioritizing rows which are close in Hamming distance and gives circuits of size at most $O(n^2/\log n)$, coinciding with the known lower bound of $\Theta(n^2/\log n)$ on the worst-case size of $\cnot$ circuits \cite{spmh02}. %We use their algorithm to perform linear reversible synthesis in our implementation.

In the realm of pure quantum circuits, Shende and Markov \cite{sm09} studied the $\cnot$ cost of Toffoli gates. They proved that $6$ $\cnot$ gates is minimal for the Toffoli gate, and gave a lower bound of $2n$ $\cnot$ gates for the $n$ qubit Toffoli gate. More recently, Welch, Greenbaum, Mostame and Aspuru-Guzik \cite{wgma14} studied the construction of efficient circuits for diagonal unitaries. They used similar insights to the ones we use here, notably the use of the $2^n$ Walsh functions as a basis for $n$-qubit diagonal operators which correspond to the parity functions in the \emph{Fourier expansions} \cite{od14} we use to describe $\cnot$-phase circuits. While their main objective was to optimize circuits by constructing approximations of the operator which use fewer Walsh functions, they give a construction of an optimal circuit computing all $2^n$ Walsh functions. They further give $\cnot$ identities which they use to optimize circuits when not all Walsh functions are used, but give no experimental data as to the effectiveness of these optimizations. In contrast, we present and test an algorithm which directly synthesizes an efficient circuit for a specific set of Walsh functions, rather than construct a circuit for the full set and optimize later.

%%%%%%%%%%%%%%%%%%%%%%%%%%%%%%%%%%%%%%%%
\section{CNOT-phase circuits and Fourier expansions}

We begin by introducing the \emph{circuit-polynomial} correspondence \cite{m17} for quantum circuits composed of $\cnot$ and $\rz$ gates, which associates a \emph{phase polynomial} and a linear Boolean transformation with every circuit. While the $\{\cnot, \rz\}$ gate set does not contain any branching gates, we call this the \emph{sum-over-paths} form of a circuit in keeping with similar analyses \cite{dhmhno05, kps17, m17}.

\begin{definition}\label{defn:sop}
The sum-over-paths form of a circuit $C$ over $\cnot$ and $\rz$ gates with associated unitary $U_C$ is given by a pseudo-Boolean function
\[
  f(\x) = \sum_{\y\in\F^n} \fourier{f}(\y) \cdot (x_1y_1\oplus x_2y_2 \oplus \cdots \oplus x_ny_n) %(-1)^{x_1^{y_1}x_2^{y_2} \cdots x_n^{y_n}}
\]
with coefficients $\fourier{f}(\y)\in\R$, together with a basis state transformation $A\in\GL{n}$ such that
\[
  U_C = \sum_{\x\in\F^n}e^{2\pi i f(\x)}\ket{A\x}\bra{\x}.
\]
In particular, $U_C$ maps the basis state $\ket{\x}$ to $e^{2\pi i f(\x)}\ket{A\x}$. We refer to the sum-over-paths form as the pair $(f, A)$.
\end{definition}

The expression of $f(\x)$ as a weighted sum of parities of $\x$ in \cref{defn:sop} is the \emph{Fourier expansion}\footnote{Standard literature (e.g. \cite{od14}) on Fourier analysis of pseudo-Boolean functions uses the multiplicative group $\{-1, 1\}$ as the Boolean group, resulting in a different set of coefficients. As the additive group $\{0,1\}$ is more natural for quantum computing and the resulting expansion obeys the properties needed, we use this representation. A full discussion can be found in \cref{app:fourier}.} of $f$ with \emph{Fourier coefficients} $\fourier{f}(\y)$ \cite{od14}. We call the collective set of Fourier coefficients the \emph{Fourier spectrum} or just the spectrum of $f$. Previous works \cite{amm14, am16, ch17} have used the term \emph{phase polynomial} to refer to the expression of $f(\x)$ over the parity basis -- we use the term Fourier expansion to avoid confusion with regular polynomial representations of $f$.

For convenience we define $\parity{\y}:\F^n \rightarrow \F$ to be the parity function for the indicator vector $\y\in\F^n$. In particular, 
\[
  \parity{\y}(\x) = x_1y_1 \oplus x_2y_2 \oplus \cdots \oplus x_ny_n.
\]
The Fourier expansion of $f$ is then written as
\[
  f(\x) =  \sum_{\y\in\F^n} \fourier{f}(\y)\parity{\y}(\x).
\]
We also define the \emph{support} of the Fourier expansion of $f$ to be the set of parities with non-zero coefficients -- that is,
\[
  \supp(\fourier{f}) =  \{\y \in\F^n \mid \fourier{f}(\y)\neq 0\}.
\]

It was previously shown by Amy, Maslov and Mosca \cite{amm14} that there exists a canonical sum-over-paths expression for every $\{\cnot, T\}$ circuit, and in particular it can be computed in time polynomial in the number of qubits. The same method may be used to compute a canonical sum-over-paths form for any circuit over $\{\cnot, \rz\}$, with the only difference being that the Fourier coefficients are defined over $\R$ rather than $\Z_8$. %Note that when we say a circuit has sum-over-paths form $(f, A)$ we mean the canonical sum-over-paths form.

\begin{proposition}\label{prop:canon}
Every circuit $C$ written over $\{\cnot, \rz\}$ has a canonical, polynomial-time computable sum-over-paths form.
\end{proposition}

Given a $\{\cnot, \rz\}$ circuit, its sum-over-paths form can be computed by first constructing the \emph{annotated circuit}, where the inputs are labelled by $x_1, x_2,\dots, x_n$, and outputs of every gate are labelled by a parity of the inputs. Note that by the linearity of quantum circuits and the fact that both $\cnot$ and $\rz$ gates map basis states to basis states (possibly with a phase), the state of each qubit can in fact be described at any point in the circuit as a parity over the input (basis) state. As the $\cnot$ gate maps $\ket{x}\ket{y}$ to $\ket{x}\ket{x\oplus y}$, the output for the control bit of a $\cnot$ gate has the same label as its input, while the target bit is labelled with the XOR of the input labels. An $\rz(\theta)$ gate does not change the basis state of the qubit, and hence the annotation is unchanged. 

Then to construct the sum-over-paths form, for every $\rz(\theta)$ gate with incoming label $\parity{\y}(\x)$ a factor of $\theta\cdot \parity{\y}(\x)$ is added to $f(\x)$ -- equivalently, $\fourier{f}(\y)$ is taken to be the sum of the parameters ($\theta$) of all $\rz$ gates with incoming label $\parity{\y}(\x)$. The linear transformation $A$ is defined by the mapping $\x \mapsto \x'$ where $\x'$ is the labels at the end of the circuit.

\begin{figure}
\centerline{\scriptsize
	\Qcircuit @C=.8em @R=.6em {
		\lstick{x_1} & \gate{\rz\left(\frac{1}{8}\right)} & \targ & 
			\ustick{\:x_1\oplus x_3}\qw & \push{\rule{0em}{1em}}\qw & 
			\gate{\rz\left(\frac{3}{8}\right)} & \targ & 
			\ustick{\;\;\;\;\;\;x_1\oplus x_2\oplus x_3} \qw & 
			\push{\rule{0em}{1em}}\qw & \push{\rule{0em}{1em}}\qw & \push{\rule{0em}{1em}}\qw & \push{\rule{0em}{1em}}\qw &
			\gate{\rz\left(\frac{1}{8}\right)} & \targ & 
			\ustick{\;\;\;x_1\oplus x_2} \qw & \push{\rule{0em}{1em}}\qw & \push{\rule{0em}{1em}}\qw &
			\gate{\rz\left(\frac{3}{8}\right)} &
			\targ & \ustick{\!\!\!\!\!\!\!x_1}\qw & \rstick{\!\!\!\! x_1} \\
		\lstick{x_2} & \gate{\rz\left(\frac{1}{8}\right)} & \qw & 
			\ctrl{1} & \push{\rule{0em}{1em}}\qw & 
			\qw & \ctrl{-1} & 
			\ctrl{1} & \qw & \push{\rule{0em}{1em}}\qw & \push{\rule{0em}{1em}}\qw & \push{\rule{0em}{1em}}\qw &
			\qw & \qw & 
			\qw & \push{\rule{0em}{1em}}\qw &\push{\rule{0em}{1em}}\qw & \qw &
			\ctrl{-1} & \qw & \rstick{\!\!\!\! x_2} \\
		\lstick{x_3} & \qw & \ctrl{-2} & 
			\targ & \ustick{\:\:\:\:\:x_2\oplus x_3}\qw & 
			\push{\rule{0em}{1em}}\qw & \gate{\rz\left(\frac{3}{8}\right)} & 
			\targ & \ustick{\!\!\!\!\!\!x_3}\qw & \push{\rule{0em}{1em}}\qw & \push{\rule{0em}{1em}}\qw & \push{\rule{0em}{1em}}\qw &
			\qw & \ctrl{-2} & \qw & \push{\rule{0em}{1em}}\qw &  \push{\rule{0em}{1em}}\qw &
			\gate{\rz\left(\frac{1}{8}\right)} & \qw & \qw & \rstick{\!\!\!\! x_3}
	}
}
\caption{An annotated circuit implementing the doubly controlled $Z$ gate.}
\label{fig:ccZ}
\end{figure}
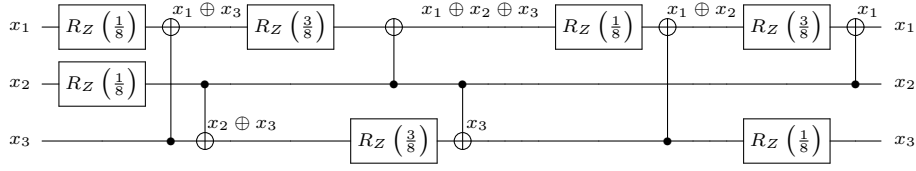

As an example, \cref{fig:ccZ} shows a circuit implementing the doubly controlled $Z$ operator
\[
  \Lambda_2(Z) = \sum_{\x\in\F^3} e^{2\pi i \frac{x_1x_2x_3}{2}}\ket{\x}\bra{\x}.
\]
The state of each qubit as a parity of an input basis state $\ket{x_1x_2x_3}$, $x_1,x_2,x_3\in\F$, is annotated after each gate. Annotations are only shown in cases where the state of the qubit is changed for clarity. Summing up the phase rotations due to $\rz$ gates gives the expression
\[
  f(\x) = \frac{1}{8}\left(x_1 + x_2 + 3(x_1\oplus x_3) + 3(x_2\oplus x_3) +
  (x_1\oplus x_2\oplus x_3) + 3(x_1\oplus x_2) + x_3 \right).
\]
As the circuit returns each qubit to its initial state, the sum-over-paths form is $(f, I)$ where $I\in\GL{3}$ is the identity matrix.

\subsection{Parity networks}

A key observation about the canonical sum-over-paths forms is that \emph{only parities which appear in the annotated circuit may have non-zero Fourier coefficient}. Otherwise, the parities may appear in any order, parities may appear in the circuit but not in the sum-over-paths form, or the same parity may appear multiple times. Multiple $\rz$ gates may also be applied with the same incoming parity throughout a circuit, in which case the Fourier coefficient is the sum of all the rotation angles and can be replaced with a single $\rz$ gate -- this effect was previously used to optimize $T$-count in Clifford+$T$ circuits by merging $T$ gates applied to the same parity \cite{amm14}.

The inverse of the above observation is that \emph{a circuit in which every parity in $\supp(\fourier{f})$ appears as an annotation can be modified to implement the phase rotation $f(\x) = \sum_{\y\in\F^n}\fourier{f}(\y)\parity{\y}(\x)$ at no extra $\cnot$ cost}. For example, the circuit in \cref{fig:ccZ} can be modified to give a new circuit with (non-equivalent) sum over paths form $(f', I)$ for 
\[
  f'(\x) = \frac{2}{3}\left(x_2\oplus x_3\right) + \frac{1}{3}\left(x_1\oplus x_2\oplus x_3\right)
\]
with no additional $\cnot$ gates simply by changing the parameters of the fourth and fifth $\rz$ gates to $\frac{2}{3}$ and $\frac{1}{3}$, respectively, and removing all other phase gates. The resulting circuit is shown below:
\[
\centerline{\scriptsize
	\Qcircuit @C=.8em @R=.6em {
		 & \targ &
			\qw & \targ &
			\gate{\rz\left(\frac{1}{3}\right)} & \targ & 
			\targ & \qw  \\
		 & \qw & 
			\ctrl{1} & \ctrl{-1} & 
			\ctrl{1} & \qw &
			\ctrl{-1} & \qw  \\
		 & \ctrl{-2} & 
			\targ & 
			\gate{\rz\left(\frac{2}{3}\right)} & 
			\targ & \ctrl{-2} & 
			\qw & \qw 
	}
}
\]
This motivates the definition of a \emph{parity network} below as a $\cnot$ circuit computing a set of parities, which can be used to implement phase rotations with Fourier expansions having support contained in that set.

\begin{definition}

A \emph{parity network} for a set $S\subseteq\F^n$ is an $n$-qubit circuit $C$ over $\cnot$ gates where, for each $\y\in S$, the parity $\parity{\y}(\x)$ appears in the annotated circuit. 
\end{definition}

As all parity networks apply some overall linear transformation of the input, we say a parity network is \emph{pointed at} $A\in\GL{n}$ if the overall transformation is $A$, i.e.,
\[
  U_{C} = \sum_{\x\in\F^n} \ket{A\x}\bra{\x}.
\]
For convenience we refer to a parity network with the trivial transformation as an \emph{identity parity network}. In the context of synthesizing parity networks, we use the term \emph{pointed parity network} to refer to a parity network applying a specific linear transformation.

We can now formalize the above observations with the following proposition, stating that the problem of finding a minimal size (pointed) parity network is equivalent to finding a $\cnot$-minimal circuit having a particular sum-over-paths form. For the remainder of the paper we consider these problems interchangeable.

\begin{proposition}\label{prop:skeleton-equiv}
Given a circuit $C$ over $\{\cnot, \rz\}$ with sum-over-paths form $(f, A)$, the circuit $C'$ obtained from $C$ by removing all phase gates is a parity network for $\supp(\fourier{f})$ pointed at $A$. 

Furthermore, given parity network $C$ for $S\subseteq\F^n$ pointed at $A$, the circuit $C'$ obtained from $C$ by, for every $\y\in S$, inserting a phase gate $\rz(\fourier{f}(\y))$ where $\parity{\y}(\x)$ appears as an annotation, has sum-over-paths form $(f, A)$.
\end{proposition}
\begin{proof}
For the former statement, by definition of the canonical sum-over-paths form the annotated version of $C$ necessarily contains $\parity{\y}(\x)$ as a label for every $\y\in\supp(\fourier{f})$. Since $\rz$ gates do not change labels or permute basis vectors, the circuit $C'$ contains all the labels of $C$ and implements the permutation $A$, hence $C'$ is a parity network $\supp(\fourier{f})$ pointed at $A$.

The proof of the latter statement is similar.
\end{proof}

%%%%%%%%%%%%%%%%%%%%%%%%%%%%%%%%%%%%%%%%
\subsection{From \texorpdfstring{$\cnot$}--minimization to parity network synthesis}

\Cref{prop:skeleton-equiv} implies that the problem of finding a minimal pointed size parity network is equivalent to the problem of finding a $\cnot$-minimal circuit for a particular sum-over-paths form. However, it is not necessarily the case that a $\cnot$-minimal circuit having a particular sum-over-paths form is a $\cnot$-minimal circuit implementing a particular unitary matrix. Since for any integer-valued function $k:\F^n\rightarrow \Z$,
\[
  e^{2\pi i f(\x)} = e^{2\pi i (f(\x) + k(\x))},
\]
it may in general be possible to instead synthesize a \emph{different} sum-over-paths form giving the same unitary operator, but with lower $\cnot$ cost. For instance,
\[
  \frac{1}{2}(x_1\oplus x_2), \qquad \text{ and }\qquad \frac{1}{2}x_1 + \frac{1}{2}x_2
\]
differ by an integer-valued function, $k(x_1,x_2)=x_1x_2$, and hence implement the same phase rotation. The left expression, together with the identity basis state transformation, gives rise to a minimal circuit containing $2$ $\cnot$ gates, while the expression on the right requires no $\cnot$ gates to implement, at the expense of an extra phase gate. The two annotated circuits are shown below.
{\scriptsize\[{
  \Qcircuit @C=.8em @R=1.4em {
  	 \lstick{x_1} & \ctrl{1} & \qw & \qw & \qw & \qw & \qw & \ctrl{1} & \qw & \qw & \rstick{\!\!\!\! x_1} \\
  	 \lstick{x_2} & \targ & \ustick{\;\;\;\;\;\;x_1\oplus x_2}\qw & \qw & \qw & \qw & 
  	 \gate{\rz\left(\frac{1}{2}\right)} & \targ & \ustick{x_2}\qw & \qw & \rstick{\!\!\!\! x_2}
  }}
  \qquad\qquad\qquad
  {
  \Qcircuit @C=.8em @R=.6em {
  	 \lstick{x_1} & \gate{\rz\left(\frac{1}{2}\right)} & \qw & \rstick{\!\!\!\! x_1} \\
  	 \lstick{x_2} & \gate{\rz\left(\frac{1}{2}\right)} & \qw & \rstick{\!\!\!\! x_2} 
  }}
\]}

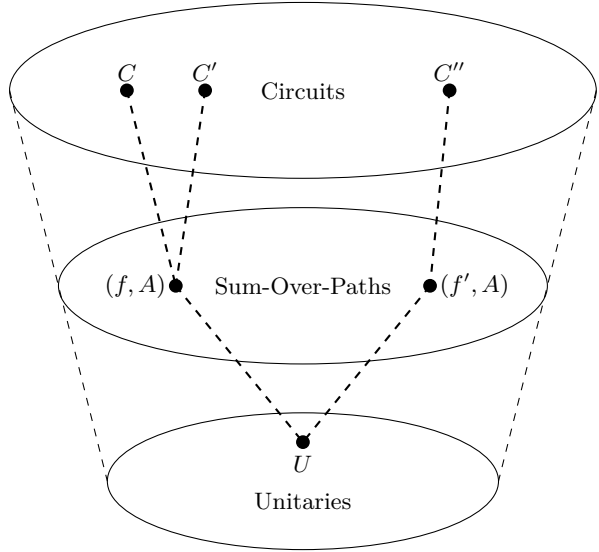
\begin{figure}
\small
\begin{tikzpicture}[>=latex,scale=1.3]
  \draw (0,0) ellipse (3 and .9) ;
  \draw (0,-2) ellipse (2.5 and .8) ;
  \draw (0,-4) ellipse (2 and .7) ;
  \draw (0,0) node{Circuits} ;
  \draw (0,-2) node{Sum-Over-Paths} ;
  \draw (0,-4.2) node{Unitaries} ;
  \fill [black] (-1.8,0) circle (2pt) ;
  \fill [black] (-1,0) circle (2pt) ;
  \fill [black] (1.5,0) circle (2pt) ;
  \fill [black] (-1.3,-2) circle (2pt) ;
  \fill [black] (1.3,-2) circle (2pt) ;
  \fill [black] (0,-3.6) circle (2pt) ;
  \draw[thick,dashed] (-1.8,0) -- (-1.3,-2) ;
  \draw[thick,dashed] (-1,0) -- (-1.3,-2) ;
  \draw[thick,dashed] (1.5,0) -- (1.3,-2) ;
  \draw[thick,dashed] (-1.3,-2) -- (0,-3.6) ;
  \draw[thick,dashed] (1.3,-2) -- (0,-3.6) ;
  \node [above] at (-1.8,0.01) {$C$} ;
  \node [above] at (-1,0.01) {$C'$} ;
  \node [above] at (1.5,0.01) {$C''$} ;
  \node [left] at (-1.31,-2) {$(f, A)$} ;
  \node [right] at (1.31,-2) {$(f', A)$} ;
  \node [below] at (0,-3.65) {$U$} ;
  \draw[dashed] (3,0) -- (2.5,-2) -- (2,-4) ;
  \draw[dashed] (-3,0) -- (-2.5,-2) -- (-2,-4) ;
\end{tikzpicture}
\caption{Relationship between circuit, sum-over-paths and unitary representations. As an example, the circuits $C, C'$ and $C''$ all correspond to the same unitary operator $U$, while only $C$ and $C'$ have the same sum-over-paths representation, $(f, A)$. In either case, the basis state transformation $A$ is the same.}
\label{fig:SOP}
\end{figure}

\Cref{fig:SOP} summarizes the relationship between circuits, sum-over-paths, and unitaries. As we are concerned with the question of minimizing $\cnot$ gates over circuits with equal \emph{unitary} representations, a natural question is how this relates to the question of minimizing $\cnot$ gates over circuits with equal \emph{sum-over-paths} representations. We now show that so long as no rotation gates have angles which are \emph{dyadic fractions} -- numbers of the form $\frac{a}{2^b}$ where $a$ and $b$ are integers -- the problems coincide.

We first formalize the intuition that two sum-over-paths forms correspond to equivalent unitaries if and only if their phases are related by an integer-valued function. In particular, we define the equivalence class of a phase function as
\[
  [f] = \{ f' : \F^n\rightarrow\R \mid f' = f + k\text{ where } k:\F^n\rightarrow\Z\},
\]
and we say $f'$ is equivalent to $f$, written $f'\sim f$, if $f' \in [f]$. From these definitions the following proposition follows straightforwardly.
\begin{proposition}\label{prop:equiv}
Given $f, f':\F^n\rightarrow \R$ and $A, A'\in\GL{n}$, the unitary matrices
\[
\sum_{\x\in\F^n}e^{2\pi i f(\x)}\ket{A\x}\bra{\x}, \qquad \sum_{\x\in\F^n}e^{2\pi i f'(\x)}\ket{A'\x}\bra{\x}
\]
are equal if and only if $A'=A$ and $f'\sim f$.
\end{proposition}

%%%%%%%%%%%%%%%%%%%%%%%%%%%%%%%%%%%%%%%%%%%%%%

It is a known fact\footnote{As we use a slightly different definition of Fourier expansion, not every pseudo-Boolean function has a unique Fourier expansion. Our Fourier expansions are only unique up to constant terms, which correspond to global phase factors and are not directly synthesizable over $\{\cnot, \rz\}$.} \cite{od14} that every pseudo-Boolean function $f:\F^n\rightarrow \R$ has a unique Fourier expansion. However, to study the relationship between Fourier expansions of equivalent functions, it will be important to know their precise form in the case of integer-valued functions.

\begin{proposition}\label{prop:groupelems}
For any integer-valued function $k:\F^n\rightarrow \Z$, the Fourier coefficients of $k$ are dyadic fractions.
\end{proposition}
\begin{proof}
Let $k:\F^n\rightarrow \Z$ be an integer-valued pseudo-Boolean function. It is known \cite{hr68} that $k$ has a unique representation as an $n$-ary multilinear polynomial over $\Z$ -- that is,
\[
  k(\x) = \sum_{\y\in\F^n} a_{\y} \x^{\y}
\]
where $\x^{\y} = x_1^{y_1}x_2^{y_2} \cdots x_n^{y_n}$ and $a_{\y}\in\Z$ for all $\y$.

Using the identity $x + y - (x \oplus y) = 2xy$ for $x, y \in \F$, we can derive an inclusion-exclusion formula for the monomial $\x^{\y}$ which we prove explicitly in \cref{app:incex}:
\begin{equation}\label{eq:incex}
  2^{|\y|-1}\x^{\y} = \sum_{\y'\subseteq \y} (-1)^{|\y'|-1}\parity{\y'}(\x).
\end{equation}
Note that binary vectors are viewed as subsets of $\{1,\dots, n\}$ for convenience. Since $a_{\y}\in\Z$ for all $\y$ and dyadic fractions are closed under addition, we observe that the Fourier coefficients of $k(\x)$ are dyadic fractions:
\[
  k(\x) = \sum_{\y\in\F^n} a_{\y} \x^{\y} = \sum_{\y\in\F^n}\left(\sum_{\y'\subseteq \y} (-1)^{|\y'|-1}\frac{a_{\y}}{2^{|\y|-1}}\right)\parity{\y'}(\x).
\]
\end{proof}

The proof of \cref{prop:groupelems} also suffices to prove a more general result, namely that any function from $\F^n$ to an Abelian group $G$ in which $2$ is a regular element has a unique Fourier expansion over $G$. This more general version also subsumes a similar result proven in \cite{am16}, namely that Fourier expansions are unique over any cyclic group of order co-prime to $2$; indeed, $2$ is regular in any such group.

We next use the above proposition to show that any pseudo-Boolean function with non-dyadic spectrum has a property of \emph{minimal support} over all equivalent functions. This is important as a parity network for $S$ is also a parity network for any subset of $S$.

\begin{proposition}\label{prop:fouriersupport}
Let $f:\F^n\rightarrow\R$ be a pseudo-Boolean function having a Fourier spectrum not containing any non-zero dyadic fractions. Then for any $f'\sim f$, 
\[
  \supp(\fourier{f})\subseteq\supp(\fourier{f'}).
\]
\end{proposition}
\begin{proof}
Consider some pseudo-Boolean function $f'$ such that $f' \sim f$. By definition we have $f' = f + k$ for some function $k:\F^n \rightarrow \Z$.  Expanding $f(\x)$ and $k(\x)$ with their Fourier expansions we have
\[
  f'(\x) = f(\x) + k(\x) = \sum_{\y\in\F^n} (\fourier{f}(\y) + \fourier{k}(\y))\parity{\y}(\x).
\]

Now since for any $\y$, $\fourier{k}(\y) = \frac{a}{2^b}$, $\fourier{f}(\y) + \fourier{k}(\y) \neq 0$. Thus $\supp(\fourier{f})\subseteq \supp(\fourier{f'})$ as required.
\end{proof}

We can now prove that the problem of synthesizing a minimal (pointed) parity network is at least as hard as general $\cnot$-minimization. As a corollary, synthesizing a minimal parity network solves the $\cnot$-minimization problem whenever the rotation angles, and hence the Fourier coefficients, are not dyadic fractions.

\begin{theorem}\label{thm:synthesishardness}
Given $A\in\GL{n}$, the problem of finding a minimal parity network for $S\subseteq \F^n$ pointed at $A$ reduces (in polynomial time) to the problem of finding a $\cnot$-minimal circuit equivalent to an $n$-qubit circuit $C$ over $\{\cnot, \rz\}$.
\end{theorem}
\begin{proof}
Given $A\in\GL{n}$ and $S\subseteq\F^n$, define $f:\F^n\rightarrow \R$ as
\[
  f(\x) = \sum_{\y\in S} \frac{1}{3}\parity{\y}(\x).
\]
It is known \cite{amm14} that a circuit $C$ over $\{\cnot, \rz\}$ implementing the sum-over-paths form $(f, A)$ can be constructed in polynomial time.

Now let $C'$ be a $\cnot$-minimal circuit equivalent to $C$. By \cref{prop:equiv} the sum-over-paths form of $C'$ must be $(f', A)$ for some $f'\sim f$. However, by \cref{prop:fouriersupport}, $S=\supp(\fourier{f})\subseteq\supp(\fourier{f'})$, so by definition, the circuit obtained from $C'$ by removing all $\rz$ gates is a (necessarily minimal) parity network for $S$ pointed at $A$.
\end{proof}

\begin{remark}
In cases when the Fourier coefficients contain dyadic fractions, it may in general be possible to further minimize $\cnot$-count by optimizing over all equivalent phase functions. This question was studied in \cite{am16} from the perspective of \emph{$T$-count} optimization -- in that case, the number of $T$ gates depends only on the size of the support, $|\supp(\fourier{f})|$. By contrast, the size of the support of a Fourier spectrum does not necessarily correspond to the size of a minimal parity network -- e.g., $\frac{1}{2}(x_1\oplus x_2)$ has smaller support than $\frac{1}{2}x_1 + \frac{1}{2}x_2$ but a larger minimal parity network as shown earlier -- which appears to make the problem of minimizing $\cnot$s size over all equivalent functions more difficult.
\end{remark}

%%%%%%%%%%%%%%%%%%%%%%%%%%%%%%%%%%%%%%%%
\section{Complexity of parity network minimization}\label{sec:complexity}

We now turn to the question of the complexity of finding minimal parity networks. We study two cases in particular where the problem can be shown to be NP-complete -- the fixed-target case, and with encoded inputs (i.e. with ancillae). At the end of the section we discuss the case of synthesizing a minimal parity network with arbitrary targets and no ancillae.

Note that we focus on the problem of synthesizing minimal parity networks rather than \emph{pointed} parity networks -- that is, synthesizing a minimal parity network up to some arbitrary overall linear transformation. However, in the cases we consider equivalent reductions for pointed parity networks are also possible.

\subsection{Fixed-target minimal parity network}

We call the problem of synthesizing a minimal parity network in which every $\cnot$ gate has the same target the \emph{fixed-target minimal parity network problem}. Formally, we define the associated decision problem MPNP$_{\text{FT}}$ below:

\smallskip
\begin{tabular}{lll}
	\hspace{1em} & {\sc Problem:} & Fixed-target minimal parity network (MPNP$_{\text{FT}}$) \\
	\hspace{1em} & {\sc Instance:} & A set of strings $S\subseteq\F^n$, and a positive integer $k$. \\
	\hspace{1em} & {\sc Question:} & Does there exist an $n$-qubit circuit $C$ over $\cnot$ gates of \\
		& & length at most $k$ such that $C$ is a parity network for $S$?
\end{tabular}
\smallskip

\begin{remark}
In general not every set of strings $S$ admits an ancilla-free fixed-target parity network, as the value of the target bit necessarily appears in every parity calculation of a fixed-target $\cnot$ circuit. It follows that an (ancilla-free) parity network for $S$ is synthesizeable if and only if there exists an index $i$ such that for every $\y\in S$, $y_i = 1$. However, a fixed-target parity network may always be synthesized by adding a single ancillary bit initialized to the state $\ket{0}$. In particular, given a set $S\subseteq\F^n$ and $A\in\GL{n}$, we may construct
\[
  S' = \{ (\y \concat 1) \mid \y \in S\},
\]
where $\y \concat 1$ denotes the length $n+1$ string obtained by concatenating $\y$ with $1$. It may then be observed that a fixed-target parity network for $S'$ is always synthesizeable, and in particular forms a parity network for $S$ when the $(n+1)$th bit is initialized to $\ket{0}$.
\end{remark}

To show that the fixed-target minimal parity network problem is NP-complete, we introduce the \emph{Hamming salesman problem} (HTSP) \cite{ernvall1985np}. Recall that the $n$-dimensional hypercube is the graph with vertices $\x\in\F^n$ and edges between $\x, \y\in\F^n$ if $\x$ and $\y$ differ in one coordinate (i.e. have Hamming distance $1$).

\smallskip
\begin{tabular}{lll}
	\hspace{1em} & {\sc Problem:} & Hamming salesman (HTSP) \\
	\hspace{1em} & {\sc Instance:} & A set of strings $S\subseteq\F^n$, and a positive integer $k$. \\
	\hspace{1em} & {\sc Question:} & Does there exist a path in the $n$-dimensional hypercube of \\
		& & length at most $k$ starting at $\0$ and going through each \\
		& & vertex $\y\in S$?
\end{tabular}
\smallskip

An equivalent (from a complexity theoretic viewpoint) version of the Hamming salesman problem exists where a cycle rather than path is found. Intuitively, the Hamming salesman problem is to find a sequence of at most $k$ bit-flips iterating through every string in some set $S$ starting from the initial string $00\dots 0$. In the case when $S=\F^n$ the minimal number of bit flips is known to be $2^n$, corresponding to one bit flip per string; this is the well known \emph{Gray code}, a total ordering on $\F^n$ where each subsequent string differs by exactly one bit. We will come back to this connection later in \cref{sec:opt} when designing a synthesis algorithm.

Ernvall, Katajainen and Penttonen \cite{ernvall1985np} show that the Hamming salesman problem is in fact NP-complete, hence we can use a reduction from HTSP to prove NP-completeness of MPNP$_{\text{FT}}$.

\begin{theorem}\label{thm:fixedtarget}
MPNP$_{\text{FT}}$ is NP-complete.
\end{theorem}
\begin{proof}
Clearly MPNP$_{\text{FT}}$ is in NP, as the state of each bit as a parity of the input values at each state in a $\cnot$ circuit is polynomial-time computable \cite{amm14}, and hence a parity network can be efficiently verified. Since HTSP is NP-complete \cite{ernvall1985np} it then suffices to show NP-hardness by reducing the Hamming salesman problem to the fixed-target minimal parity network problem.

Given an instance $(S\subseteq\F^n, k)$ of HTSP, we construct an instance $(S'\subseteq\F^{n+1}, k')$ of MPNP$_{\text{FT}}$ with size polynomial in $|S|\cdot n$ as follows:
\[
  S' = \{ (\x \concat 1) \mid \x \in S\}, \qquad k' = k.
\]

Suppose there exists a fixed-target parity network $C$ for $S'$ with length at most $k$. Without loss of generality we may assume that the fixed target is the $(n+1)$th bit, as if some $i\neq n+1$ is the target index, then for all $\y\in S'$, $y_i = 1=y_{n+1}$ and hence swapping bits $i$ and $n+1$ yields a parity network for $S'$. We can then construct a length $k$ hypercube path through each vertex $\y\in S$ with starting point $\0$ by mapping $C$ to a sequence of bit flips on each $\cnot$'s control bit. Indeed, by noting that
\[
  \cnot\ket{x_i}\ket{x_{n+1} \oplus \parity{\y}(\x)} = \ket{x_i}\ket{x_{n+1} \oplus \parity{\y\oplus\{i\}}(\x)},
\]
where $\{i\}$ is taken as the bitstring $\z$ with $z_i = 1, z_{j\neq i} = 0$, each $\cnot$ gate in $C$ with control $i$ has the affect of flipping the $i$th bit of $\y$. By the definition of a parity network, for every $\y\in S$, the parity
\[
  x_{n+1} \oplus \parity{\y}(\x)
\]
appears as an annotation in the circuit, in particular on the $(n+1)$th bit which had initial state $x_{n+1} \oplus \parity{\0}(\x)$,
hence the sequence of bit flips passes through each vertex $\y\in S$ starting from $\0$.

Likewise, if there exists a length $k$ tour through each $\y\in S$, given by a sequence of bit flips, the circuit defined by mapping each bit flip on $i$ to a $\cnot$ with control $i$ and target $n+1$ is a length $k$ parity network for $S'$ and $A'$.
\end{proof}

As the minimum $k$ for which a parity network exists is at most $(n-1)\cdot |S|$, the optimization version of MPNP$_{\text{FT}}$ is also in NP, and hence is NP-complete.

\begin{corollary}
The problem of finding a minimal fixed-target parity network is NP-complete.
\end{corollary}

It may be observed that the proof of \cref{thm:fixedtarget} can be modified to show that the problem of finding a minimal \emph{pointed} parity network with fixed $\cnot$ targets is also NP-complete. In particular, taking $A$ to be the identity transformation gives a reduction from the cycle version of HTSP.

\subsection{Minimal parity network with encoded inputs}

As $\cnot$-phase circuits form a relatively small \cite{acr17}, classically simulable group, one of the main applications of their optimization is to optimize sub-circuits of circuits over more powerful gate sets. The $T$-par optimization algorithm \cite{amm14} previously took this approach, re-synthesizing $\{\cnot, T\}$ subcircuits of Clifford+$T$ circuits with optimal $T$-depth. As Selinger \cite{s13} showed, adding ancillae increases the amount of $T$ gate parallelization possible when performing this re-synthesis.

In the case of $\cnot$ optimization, the situation is similar in that ancillae can have a significant effect on the number of $\cnot$ gates required to implement a parity network, particularly if the ancillae are initialized in linear combinations of the primary inputs. For instance, the boxed $\cnot$-phase sub-circuit below on the left performs a phase rotation of $\frac{1}{8}(x_2\oplus x_3)$ -- by noting that the ancilla begins the sub-circuit already in the state $x_2\oplus x_3$ we can remove both $\cnot$ gates, as shown by the equivalent circuit on the right.
{\scriptsize
\[
	\Qcircuit @C=.8em @R=.3em @!R {
  		\lstick{x_1} & \qw & \qw & \targ & \gate{H} & \qw & \qw & \qw & \qw & \rstick{\!\!\!\!x_1'} \\
  		\lstick{x_2} & \qw & \ctrl{2} & \qw & \qw & \ctrl{1} & \qw & \ctrl{1} & \qw & \rstick{\!\!\!\!x_2} \\
  		\lstick{x_3} & \ctrl{1} & \qw & \qw & \qw & \targ & \gate{T} & \targ & \qw & \rstick{\!\!\!\!x_3} \\
  		\lstick{0} & \targ & \targ & \ctrl{-3} & \qw & \qw & \qw & \qw & \qw \gategroup{1}{6}{4}{8}{1.5em}{..}  & \rstick{\!\!\!\!x_2\oplus x_3}
	}
	  \qquad\qquad\qquad\qquad\qquad\qquad
	  \Qcircuit @C=.8em @R=.3em @!R {
  		\lstick{x_1} & \qw & \qw & \targ & \gate{H} & \qw & \qw & \rstick{\!\!\!\!x_1'} \\
  		\lstick{x_2} & \qw & \ctrl{2} & \qw & \qw & \qw & \qw & \rstick{\!\!\!\!x_2} \\
  		\lstick{x_3} & \ctrl{1} & \qw & \qw & \qw & \qw & \qw & \rstick{\!\!\!\!x_3} \\
  		\lstick{0} & \targ & \targ & \ctrl{-3} & \qw & \gate{T} & \qw \gategroup{1}{6}{4}{6}{.8em}{..}  & \rstick{\!\!\!\!x_2\oplus x_3}
	}
\]
}

We now consider the problem of synthesizing minimal parity networks when some of the inputs are linear combinations of others. Formally, given a linear transformation $E\in\F^{m\times n}$ where $m>n$, a string $\w\in\F^m$ we say is an \emph{encoding} of $\x\in\F^n$ if $E\x = \w$. The \emph{minimal parity network with encoded inputs problem} (MPNP$_{\text{E}}$) is then to find a parity network for a given set $S\subseteq\F^n$ as a function of the primary inputs $\x\in\F^n$, but with the initial state $\ket{E\x}$ rather than $\ket{\x}$.

\smallskip
\begin{tabular}{lll}
	\hspace{1em} & {\sc Problem:} & Minimal parity network with encoded inputs (MPNP$_{\text{E}}$) \\
	\hspace{1em} & {\sc Instance:} & A set of strings $S\subseteq\F^n$, a linear transformation $E\in\F^{m\times n}$, \\
		& & and a positive integer $k$. \\
	\hspace{1em} & {\sc Question:} & Does there exist an $m$-qubit circuit $C$ over $\cnot$ gates of \\
		& & length at most $k$ such that $C$ is a parity network for some \\
		& & set $S'\subseteq\F^m$ where for any $\y\in S$ there exists $\w\in S'$ such \\
		& & that $E^T\w = \y$?
\end{tabular}
\smallskip

It can be observed that a parity network for some set $S'$ as above is equivalent to a parity network for $S$ starting from the initial state $\ket{E\x}$ for any $\x\in\F^n$. In particular, for any $\w\in S'$ and $\x\in\F^n$,
\[
	\parity{\w}(E\x) = \w^TE\x = \parity{E^T\w}(\x) = \parity{\y}(\x).
\]

MPNP$_{\text{E}}$ is again NP-complete, which we prove by a reduction from the well known NP-complete \emph{Maximum-likelihood Decoding Problem} (MLDP). While the focus of this paper is on $\{\cnot, \rz\}$ circuits, the proof may be modified to establish the more general result that synthesizing a $\cnot$ circuit implementing $A\in\F^{m\times n}$ is NP-complete.

\smallskip
\begin{tabular}{lll}
	\hspace{1em} & {\sc Problem:} & Maximum-likelihood decoding (MLDP) \\ 
	\hspace{1em} & {\sc Instance:} & A linear transformation $H\in\F^{m\times n}$, a vector $\y\in\F^n$, and a \\
		& & positive integer $k$. \\
	\hspace{1em} & {\sc Question:} & Does there exist a vector $\w\in\F^m$ of weight at most $k$ such \\
		& & that $H\w = \y$?
\end{tabular}
\smallskip

In the case when $H$ is the parity check matrix of a code $C$ and $\y$ is the syndrome of some vector $\z$, finding the minimum such $\w$ gives the minimum weight vector in the coset of $\z + C$, corresponding to a minimum distance decoding of $\z$. Berlekamp, McEliece, and van Tilborg \cite{bmt78} proved that the Maximum-likelihood decoding problem is NP-complete, and so we may reduce it to the minimal parity network with encoded inputs problem to show NP-completeness.

\begin{theorem}
MPNP$_{\text{E}}$ is NP-complete.
\end{theorem}
\begin{proof}
As noted in the proof of \cref{thm:fixedtarget}, MPNP$_{\text{E}}$ is clearly in NP since a parity network can be verified in polynomial time. To establish NP-hardness we give a reduction from MLDP.

Given an instance $(H, \y, k)$ of MLDP we construct an instance $(S, k')$ of MPNP$_{\text{E}}$ as follows:
\[
  S = \{\y\}, \qquad E = H^T, \qquad k' = k-1.
\]

Suppose there exists a vector $\w\in\F^m$ of weight at most $k$ such that $H\w = \y$. Then we know $S'=\{\w\}$ satisfies the requirement that for any $\y\in S$ there exists $\w\in S'$ such that $E^T\w = H\w = \y$. Moreover, the parity computation $\parity{\w}$ can be computed with $|\w| - 1\leq k'$ $\cnot$ gates, hence there exists a parity network of length at most $k'$ for $S'$

On the other hand, suppose there exists a length $\leq k'$ parity network for some set $S'$ where there exists $\w\in S'$ such that $E^T\w = H\w = \y$. By noting that
\[
  \cnot\ket{\parity{\y}(\x)}\ket{\parity{\z}(\x)} = \ket{\parity{\y}(\x)}\ket{\parity{\y\oplus\z}(\x)},
\]
$|\y\oplus\z|\leq |\y| + |\z|$. As each bit starts in some state $x_i = \parity{\{x_i\}}(\x)$ we see that the size of the parity in any bit at any point in the parity network is at most $k'+1\leq k$, and so $|\w|\leq k$ as required.
\end{proof}

\begin{corollary}
The problem of finding a minimal parity network with encoded inputs is NP-hard.
\end{corollary}

As in the fixed-target case, the problem of finding a minimal pointed parity network with encoded inputs is also NP-complete, by virtue of the fact that a minimal identity parity network for a singleton set $S=\{\y\}$ necessarily has the form $C \append C'$ where both $C$ and $(C')^{-1}$ are both parity network for $S$. Recall that the inverse circuit $(C')^{-1}$ has the same length as $C'$.

\subsection{Discussion} While the case of parity network synthesis with encoded inputs corresponds to the practical cases of synthesis with ancillae and sub-circuit re-synthesis, it relies on the hardness of finding a minimal sum of linearly dependent vectors, or minimum distance decoding. It would appear that the problem becomes easier when using unencoded inputs, as each vector and hence parity may be expressed uniquely over the inputs. We leave the complexity of finding a minimal parity network without ancillae as an open problem, but conjecture that the identity version is at least as hard as synthesizing a minimal fixed-target identity parity network. In particular, it appears that the two problems coincide whenever one bit appears in every parity.

%%%%%%%%%%%%%%%%%%%%%%%%%%%%%%%%%%%%%%%%
\section{A heuristic synthesis algorithm}\label{sec:opt}

In this section we present an efficient, heuristic algorithm for synthesizing small parity networks. The algorithm is inspired by Gray codes, which were noted in \cref{sec:complexity} to iterate through all $2^n$ elements of $\F^n$ minimally with one bit flip per string. The situation is different for synthesizing parity networks as the bits which can be flipped depend on the state of all $n$ bits, so our method works by trying to find subsets of $S$ which can be efficiently iterated with a Gray code on a fixed target. In the limit where $S=\F^n$, the algorithm gives a minimal size parity network for $S$. Again we focus just on the problem of synthesizing a parity network up to some arbitrary overall linear transformation.

\begin{algorithm}
\caption{Algorithm for synthesizing a parity network}
\label{alg:graysynth}
\begin{algorithmic}[1]
\Function{gray-synth}{$S\subseteq \F^n$}
	\State New empty circuit $C$
	\State New empty stack $Q$
	\State $Q$.push($S$, $\{1,\dots, n\}$, $\epsilon$)
	\While{$Q$ non-empty}
		\State ($S$, $I$, $i$) $\gets$ $Q$.pop
		\If{$S = \emptyset$ \textbf{or} $I=\emptyset$}
			\Return
		\ElsIf{$i\in\N$}
			\While{$\exists j\neq i \in \{1,\dots, n\}$ s.t. $\y_j = 1$ for all $\y\in S$}
				\State $C \gets C \append \cnot_{j, i}$
				\ForAll{($S'$, $I'$, $i'$)$\in Q\cup$($S$, $I$, $i$)}
					\ForAll{$\y\in S'$}
						\State $\y\gets E_{i,j}\y$
					\EndFor
				\EndFor
			\EndWhile
		\EndIf
		\State $j \gets \argmax_{j\in I} \max_{x\in\F} |\{\y\in S \mid y_j = x\}|$
		\State $S_{0} \gets \{\y \in S \mid y_j = 0\}$
		\State $S_{1} \gets \{\y \in S \mid y_j = 1\}$
		\If{$i \in \{\epsilon\}$}
			\State $Q$.push($S_{1}$, $I\setminus\{j\}$, $j$)
		\Else
			\State $Q$.push($S_{1}$, $I\setminus\{j\}$, $i$)
		\EndIf
		\State $Q$.push($S_{0}$, $I\setminus\{j\}$, $i$)
	\EndWhile
	\State \Return $C$
\EndFunction
\end{algorithmic}
\end{algorithm}

The algorithm {\sc gray-synth}, is presented in pseudo-code in \cref{alg:graysynth}. Note that $\cnot_{i, j}$ denotes a $\cnot$ gate with control $i$ and target $j$, and $E_{i, j}$ denotes the elementary $\F$-matrix adding row $i$ to row $j$. Given a set of binary strings $S$, the algorithm synthesizes a parity network for $S$ by repeatedly choosing an index $i$ to expand and then effectively recurring on the co-factors $S_0$ and $S_1$, consisting of the strings $\y\in S$ with $y_i = 0$ or $1$, respectively. As a subset $S$ is recursively expanded, $\cnot$ gates are applied so that a designated \emph{target} bit contains the (partial) parity $\parity{\y}(\x)$ where $y_i = 1$ if and only if $y'_i = 1$ for all $\y'\in S$ -- if $S$ is a singleton $\{\y'\}$, then $\y = \y'$, hence the target bit contains the value $\parity{\y'}(\x)$ as desired. Notably, rather than \emph{uncomputing} this sequence of $\cnot$ gates when a subset $S$ is finished being synthesized, the algorithm maintains the invariant that the remaining parities to be computed are expressed over the current state of the bits. This allows the algorithm to avoid the ``backtracking'' inherent in uncomputing-based methods.

More precisely, the invariant of \cref{alg:graysynth} described above is expressed in the following lemma.
\begin{lemma}\label{lem:inv}
Let $C$ be a $\cnot$ circuit and $S\subseteq \F^n$. For any positive integer $i$ we let $C_{\leq i}$ denote the first $i$ gates of $C$, $c_i$ and $t_i$ be the control and target of the $i$th $\cnot$ gate in $C$. If we define
\begin{align*} 
  A_0 &= I &\y_0 &= \y \\
  A_i &= E_{c_i, t_i}A_{i-1} &\y_i &= E_{t_i, c_i}\y_{i-1}
\end{align*}
for every $\y\in S$, then it follows that for any $\x\in\F^n$, $U_{C\leq i}\ket{\x} = \ket{A_i\x}$ and
\begin{align*}
  \parity{\y_i}(A_i\x) &= \parity{\y}(\x).
\end{align*}
\end{lemma}
\begin{proof}
The fact that $U_{C\leq i}\ket{\x} = \ket{A_i\x}$ follows simply from the fact that $\cnot_{i, j}\ket{\x} = \ket{E_{i, j}\x}$. For the latter fact, clearly $\parity{\y_i}(A_i\x) = \parity{\y}(\x)$ by definition. Moreover, recall that
\[
  \parity{\y}(A\x) = \y^TA\x = \parity{A^T\y}(\x)
\]
and hence by induction,
\begin{align*}
  \parity{\y_i}(A_i\x)
  	&= \parity{(A_i)^T\y_i}(\x) \\
  	&= \parity{(A_{i-1})^TE_{t_i, c_i}E_{t_i, c_i}\y_{i-1}}(\x) \\
	&= \parity{(A_{i-1})^T\y_{i-1}}(\x) \\
  	&= \parity{\y}(\x).
\end{align*}
\end{proof}
It is clear to see that each $\y_i$ for $\y\in S$ in \cref{lem:inv} is the value of $\y$ after $i$ iterations. To see then that the output is in fact a parity network for $S$, it suffices to observe that whenever $S= \{\y_i\}$ and $I=\emptyset$, $|\y_i| = 1$ and thus by \cref{lem:inv}, some bit is in the state $\ket{\parity{\y_i}(A_i\x)} = \ket{\parity{\y}(\x)}$ after the $i$th $\cnot$ gate. While the fact that $|\y_i| = 1$ is assured by lines 9-16 in this case, the non-zero elements of the target strings are actually zero-ed out earlier, as the algorithm expands each coordinate. In particular, the first time the ``1'' branch is taken when expanding a set, corresponding to the first $1$ seen over the indices previously examined, the target bit $i$ is set -- taking further ``1'' branches result in the row $j$ being flipped to $0$ with a single $\cnot$. In this way the algorithm makes use of the redundancy in Fourier spectrum $S$.

In practice, a parity network implementing some particular basis state transformation is typically needed. We take the approach of synthesizing pointed parity network by first synthesizing a regular parity network, then implementing the remaining linear transformation -- i.e. $AA_i^{-1}$ where $A_i$ is the linear transformation implemented by the network. In our implementation we use the Patel-Markov-Hayes algorithm \cite{pmh08} which gives asymptotically optimal $\cnot$ count.

While the correctness of \cref{alg:graysynth} is independent of the choice of index $j$ to expand in line 18, in practice it has a large impact on the size of the resulting parity network. We chose $j$ so as to maximize the size of the largest subset, $S_0$ or $S_1$, i.e. $j = \argmax_{j\in I} \max_{x\in\F} |\{\y\in S \mid y_j= x\}|$. The intuition behind this choice is that as a subset $S$ of $\F^n$ with $m$ bits fixed approaches $|S| = 2^{n-m}$, the minimal parity network for $S$ approaches one $\cnot$ per string, corresponding to the Gray code in the limit. We also ran experiments with other methods of choosing $j$; we found that $j = \argmax_{j\in I} \max_{x\in\F} |\{\y\in S \mid y_j = x\}|$ gave the best results on average.

\subsection{Examples}

\begin{example}
To illustrate \cref{alg:graysynth}, we demonstrate the use of {\sc gray-synth} to synthesize a circuit over $\{\cnot, \rz\}$ implementing the diagonal operator $U\ket{\x} = e^{2\pi i f(\x)}\ket{\x}$ given by
\[
  f(\x) = \frac{1}{8}\left[(x_2\!\oplus\! x_3) + x_1 + (x_1\!\oplus\! x_4) + (x_1\!\oplus\! x_2 \!\oplus\! x_3) + (x_1\!\oplus\! x_2 \!\oplus\! x_4) 
  	+ (x_1 \!\oplus\! x_2)\right].
\]

\hspace*{-8em}
 \begin{minipage}{0.8\textwidth}
\[
\begin{tikzpicture}
  \matrix (m)[
    matrix of math nodes,
    left delimiter  = (,
    right delimiter = ),
  ] {
0&1&1&1&1&1 \\
1&0&0&1&1&1 \\
1&0&0&1&0&0 \\
0&0&1&0&1&0 \\
  } ;

  %\draw (m-1-1.north west) rectangle (m-4-6.south east);

\end{tikzpicture}
\]
\end{minipage}
\begin{minipage}{0.6\textwidth}\footnotesize
  \Qcircuit @C=.8em @R=1.2em {
  	\lstick{x_1}& \qw & \rstick{\!\!\!\!x_1} \\
  	\lstick{x_2}& \qw & \rstick{\!\!\!\!x_2} \\
  	\lstick{x_3}& \qw & \rstick{\!\!\!\!x_3} \\
  	\lstick{x_4}& \qw & \rstick{\!\!\!\!x_4}
  }
\end{minipage}

Starting with the initial set $S$, written as the columns of the matrix on the left, we choose a bit maximizing the number of $0$'s or $1$'s in $S$. As $j=1$ in this case, we construct the cofactors $S_0$ and $S_1$ on the values in the first row, and recurse on $S_0$.

%%%%%%%%%%%%%%%%%%%%%%%%%%%%%%%%%%

\hspace*{-8em}
 \begin{minipage}{0.8\textwidth}
\[
\begin{tikzpicture}
  \matrix (m)[
    matrix of math nodes,
    left delimiter  = (,
    right delimiter = ),
  ] {
\node[node style bb]{0};&\node[node style bb]{1};&\node[node style bb]{1};&\node[node style bb]{1};&\node[node style bb]{1};&\node[node style bb]{1}; \\
1&0&0&1&1&1 \\
1&0&0&1&0&0 \\
0&0&1&0&1&0 \\
  } ;
  
  \draw (m-1-1.north west) rectangle (m-4-1.south east);

\end{tikzpicture}
\]
\end{minipage}
\begin{minipage}{0.6\textwidth}\footnotesize
  \Qcircuit @C=.8em @R=1.2em {
  	\lstick{x_1}& \qw & \rstick{\!\!\!\!x_1} \\
  	\lstick{x_2}& \qw & \rstick{\!\!\!\!x_2} \\
  	\lstick{x_3}& \qw & \rstick{\!\!\!\!x_3} \\
  	\lstick{x_4}& \qw & \rstick{\!\!\!\!x_4}
  }
\end{minipage}

%%%%%%%%%%%%%%%%%%%%%%%%%%%%%%

The greyed row above indicates that the first row has been partitioned, and the box indicates the current subset being synthesized. The algorithm next selects row $2$ and immediately descends into the $1$-cofactor, since $S_0=\emptyset$. Again, the algorithm selects row $3$, and since both rows $2$ and $3$ have the value $1$, a $\cnot$ is applied with bit $2$ as the target and $3$ as the control. The remaining vectors are updated by multiplying with $E_{2,3}$ -- the modified entries are shown in red.

\hspace*{-8em}
 \begin{minipage}{0.8\textwidth}
\[
\begin{tikzpicture}
  \matrix (m)[
    matrix of math nodes,
    left delimiter  = (,
    right delimiter = ),
  ] {
\node[node style bb]{0};&\node[node style bb]{1};&\node[node style bb]{1};&\node[node style bb]{1};&\node[node style bb]{1};&\node[node style bb]{1}; \\
\node[node style bb]{1};&0&0&1&1&1 \\
1&0&0&1&0&0 \\
0&0&1&0&1&0 \\
  } ;
  
  \draw (m-1-1.north west) rectangle (m-4-1.south east);

\end{tikzpicture}
\]
\end{minipage}
\begin{minipage}{0.6\textwidth}\footnotesize
  \Qcircuit @C=.8em @R=1.2em {
  	\lstick{x_1}& \qw & \rstick{\!\!\!\!x_1} \\
  	\lstick{x_2}& \qw & \rstick{\!\!\!\!x_2} \\
  	\lstick{x_3}& \qw & \rstick{\!\!\!\!x_3} \\
  	\lstick{x_4}& \qw & \rstick{\!\!\!\!x_4}
  }
\end{minipage}

%%%%%%%%%%%%%%%%%%%%%%%%%%%%%%

\hspace*{-8em}
 \begin{minipage}{0.8\textwidth}
\[
\begin{tikzpicture}
  \matrix (m)[
    matrix of math nodes,
    left delimiter  = (,
    right delimiter = ),
  ] {
\node[node style bb]{0};&\node[node style bb]{1};&\node[node style bb]{1};&\node[node style bb]{1};&\node[node style bb]{1};&\node[node style bb]{1}; \\
\node[node style bb]{1};&0&0&1&1&1 \\
\node[node style bb]{1};&0&0&1&0&0 \\
0&0&1&0&1&0 \\
  } ;
  
  \draw (m-1-1.north west) rectangle (m-4-1.south east);

\end{tikzpicture}
\raisebox{2.7em}{$\rightarrow$} 
\begin{tikzpicture}
  \matrix (m)[
    matrix of math nodes,
    left delimiter  = (,
    right delimiter = ),
  ] {
\node[node style bb]{0};&\node[node style bb]{1};&\node[node style bb]{1};&\node[node style bb]{1};&\node[node style bb]{1};&\node[node style bb]{1}; \\
\node[node style bb]{1};&0&0&1&1&1 \\
\node[node style ba]{0};&0&0&\node[node style aa]{0};&\node[node style aa]{1};&\node[node style aa]{1}; \\
0&0&1&0&1&0 \\
  } ;
  
  \draw (m-1-1.north west) rectangle (m-4-1.south east);

\end{tikzpicture}
\]
\end{minipage}
\begin{minipage}{0.6\textwidth}\footnotesize
  \Qcircuit @C=.8em @R=1.2em {
  	\lstick{x_1}& \qw & \qw & \rstick{\!\!\!\!x_1} \\
  	\lstick{x_2}& \targ & \qw & \rstick{\!\!\!\!x_2\oplus x_3} \\
  	\lstick{x_3}& \ctrl{-1} & \qw & \rstick{\!\!\!\!x_3} \\
  	\lstick{x_4}& \qw & \qw & \rstick{\!\!\!\!x_4}
  }
\end{minipage}

%%%%%%%%%%%%%%%%%%%%%%%%%%%%%%

As the final row has the value $0$, we're finished with this column and may continue with the remaining subsets. The single $1$ in row $2$ indicates that the second qubit currently holds the value of the corresponding parity $x_2\oplus x_3$, as seen in the circuit on the right.

\hspace*{-8em}
 \begin{minipage}{0.8\textwidth}
\[
\begin{tikzpicture}
  \matrix (m)[
    matrix of math nodes,
    left delimiter  = (,
    right delimiter = ),
  ] {
\node[node style bb]{0};&\node[node style bb]{1};&\node[node style bb]{1};&\node[node style bb]{1};&\node[node style bb]{1};&\node[node style bb]{1}; \\
\node[node style bb]{1};&0&0&1&1&1 \\
\node[node style bb]{0};&0&0&0&1&1 \\
\node[node style bb]{0};&0&1&0&1&0 \\
  } ;
  
  \draw (m-1-1.north west) rectangle (m-4-1.south east);

\end{tikzpicture}
\]
\end{minipage}
\begin{minipage}{0.6\textwidth}\footnotesize
  \Qcircuit @C=.8em @R=1.2em {
  	\lstick{x_1}& \qw & \qw & \rstick{\!\!\!\!x_1} \\
  	\lstick{x_2}& \targ & \qw & \rstick{\!\!\!\!x_2\oplus x_3} \\
  	\lstick{x_3}& \ctrl{-1} & \qw & \rstick{\!\!\!\!x_3} \\
  	\lstick{x_4}& \qw & \qw & \rstick{\!\!\!\!x_4}
  }
\end{minipage}

%%%%%%%%%%%%%%%%%%%%%%%%%%%%%%

Again the algorithm chooses row $2$ maximizing the number of entries which are the same for the remaining columns, and recurses on the $0$-cofactor as shown below.

\hspace*{-8em}
 \begin{minipage}{0.8\textwidth}
\[
\begin{tikzpicture}
  \matrix (m)[
    matrix of math nodes,
    left delimiter  = (,
    right delimiter = ),
  ] {
\node[node style bb]{0};&\node[node style bb]{1};&\node[node style bb]{1};&\node[node style bb]{1};&\node[node style bb]{1};&\node[node style bb]{1}; \\
\node[node style bb]{1};&\node[node style bb]{0};&\node[node style bb]{0};&\node[node style bb]{1};&\node[node style bb]{1};&\node[node style bb]{1}; \\
\node[node style bb]{0};&0&0&0&1&1 \\
\node[node style bb]{0};&0&1&0&1&0 \\
  } ;
  
  \draw (m-1-2.north west) rectangle (m-4-3.south east);

\end{tikzpicture}
\]
\end{minipage}
\begin{minipage}{0.6\textwidth}\footnotesize
  \Qcircuit @C=.8em @R=1.2em {
  	\lstick{x_1}& \qw & \qw & \rstick{\!\!\!\!x_1} \\
  	\lstick{x_2}& \targ & \qw & \rstick{\!\!\!\!x_2\oplus x_3} \\
  	\lstick{x_3}& \ctrl{-1} & \qw & \rstick{\!\!\!\!x_3} \\
  	\lstick{x_4}& \qw & \qw & \rstick{\!\!\!\!x_4}
  }
\end{minipage}

%%%%%%%%%%%%%%%%%%%%%%%%%%%%%%

\hspace*{-8em}
 \begin{minipage}{0.8\textwidth}
\[
\begin{tikzpicture}
  \matrix (m)[
    matrix of math nodes,
    left delimiter  = (,
    right delimiter = ),
  ] {
\node[node style bb]{0};&\node[node style bb]{1};&\node[node style bb]{1};&\node[node style bb]{1};&\node[node style bb]{1};&\node[node style bb]{1}; \\
\node[node style bb]{1};&\node[node style bb]{0};&\node[node style bb]{0};&\node[node style bb]{1};&\node[node style bb]{1};&\node[node style bb]{1}; \\
\node[node style bb]{0};&\node[node style bb]{0};&\node[node style bb]{0};&0&1&1 \\
\node[node style bb]{0};&0&1&0&1&0 \\
  } ;
  
  \draw (m-1-2.north west) rectangle (m-4-3.south east);

\end{tikzpicture}
\]
\end{minipage}
\begin{minipage}{0.6\textwidth}\footnotesize
  \Qcircuit @C=.8em @R=1.2em {
  	\lstick{x_1}& \qw & \qw & \rstick{\!\!\!\!x_1} \\
  	\lstick{x_2}& \targ & \qw & \rstick{\!\!\!\!x_2\oplus x_3} \\
  	\lstick{x_3}& \ctrl{-1} & \qw & \rstick{\!\!\!\!x_3} \\
  	\lstick{x_4}& \qw & \qw & \rstick{\!\!\!\!x_4}
  }
\end{minipage}

%%%%%%%%%%%%%%%%%%%%%%%%%%%%%%

In expanding the last row, we first examine the $0$-cofactor and find nothing to do, then the $1$-cofactor, at which point we need to apply a $\cnot$ with target bit $1$ and control $4$.

\hspace*{-8em}
 \begin{minipage}{0.8\textwidth}
\[
\begin{tikzpicture}
  \matrix (m)[
    matrix of math nodes,
    left delimiter  = (,
    right delimiter = ),
  ] {
\node[node style bb]{0};&\node[node style bb]{1};&\node[node style bb]{1};&\node[node style bb]{1};&\node[node style bb]{1};&\node[node style bb]{1}; \\
\node[node style bb]{1};&\node[node style bb]{0};&\node[node style bb]{0};&\node[node style bb]{1};&\node[node style bb]{1};&\node[node style bb]{1}; \\
\node[node style bb]{0};&\node[node style bb]{0};&\node[node style bb]{0};&0&1&1 \\
\node[node style bb]{0};&\node[node style bb]{0};&\node[node style bb]{1};&0&1&0 \\
  } ;

  \draw (m-1-2.north west) rectangle (m-4-2.south east);

\end{tikzpicture}
\]
\end{minipage}
\begin{minipage}{0.6\textwidth}\footnotesize
  \Qcircuit @C=.8em @R=1.2em {
  	\lstick{x_1}& \qw & \qw & \rstick{\!\!\!\!x_1} \\
  	\lstick{x_2}& \targ & \qw & \rstick{\!\!\!\!x_2\oplus x_3} \\
  	\lstick{x_3}& \ctrl{-1} & \qw & \rstick{\!\!\!\!x_3} \\
  	\lstick{x_4}& \qw & \qw & \rstick{\!\!\!\!x_4}
  }
\end{minipage}

%%%%%%%%%%%%%%%%%%%%%%%%%%%%%%

\hspace*{-8em}
 \begin{minipage}{0.8\textwidth}
\[
\begin{tikzpicture}
  \matrix (m)[
    matrix of math nodes,
    left delimiter  = (,
    right delimiter = ),
  ] {
\node[node style bb]{0};&\node[node style bb]{1};&\node[node style bb]{1};&\node[node style bb]{1};&\node[node style bb]{1};&\node[node style bb]{1}; \\
\node[node style bb]{1};&\node[node style bb]{0};&\node[node style bb]{0};&\node[node style bb]{1};&\node[node style bb]{1};&\node[node style bb]{1}; \\
\node[node style bb]{0};&\node[node style bb]{0};&\node[node style bb]{0};&0&1&1 \\
\node[node style bb]{0};&\node[node style bb]{0};&\node[node style bb]{1};&0&1&0 \\
  } ;

  \draw (m-1-3.north west) rectangle (m-4-3.south east);

\end{tikzpicture}
\raisebox{2.7em}{$\rightarrow$} 
\begin{tikzpicture}
  \matrix (m)[
    matrix of math nodes,
    left delimiter  = (,
    right delimiter = ),
  ] {
\node[node style bb]{0};&\node[node style bb]{1};&\node[node style bb]{1};&\node[node style bb]{1};&\node[node style bb]{1};&\node[node style bb]{1}; \\
\node[node style bb]{1};&\node[node style bb]{0};&\node[node style bb]{0};&\node[node style bb]{1};&\node[node style bb]{1};&\node[node style bb]{1}; \\
\node[node style bb]{0};&\node[node style bb]{0};&\node[node style bb]{0};&0&1&1 \\
\node[node style bb]{0};&\node[node style bb]{0};&\node[node style ba]{0};&\node[node style aa]{1};&\node[node style aa]{0};&\node[node style aa]{1}; \\
  } ;

  \draw (m-1-3.north west) rectangle (m-4-3.south east);

\end{tikzpicture}
\]
\end{minipage}
\begin{minipage}{0.6\textwidth}\footnotesize
  \Qcircuit @C=.8em @R=1em {
  	\lstick{x_1}& \qw & \targ & \qw & \rstick{\!\!\!\!x_1\oplus x_4} \\
  	\lstick{x_2}& \targ & \qw & \qw & \rstick{\!\!\!\!x_2\oplus x_3} \\
  	\lstick{x_3}& \ctrl{-1} & \qw & \qw & \rstick{\!\!\!\!x_3} \\
  	\lstick{x_4}& \qw & \ctrl{-3} & \qw & \rstick{\!\!\!\!x_4}
  }
\end{minipage}

%%%%%%%%%%%%%%%%%%%%%%%%%%%%%%

Now backtracking and entering the $1$-cofactor for the remaining columns, we find we need to apply a $\cnot$ between bits $2$ and $1$ to zero out the row.

\hspace*{-8em}
 \begin{minipage}{0.8\textwidth}
\[
\begin{tikzpicture}
  \matrix (m)[
    matrix of math nodes,
    left delimiter  = (,
    right delimiter = ),
  ] {
\node[node style bb]{0};&\node[node style bb]{1};&\node[node style bb]{1};&\node[node style bb]{1};&\node[node style bb]{1};&\node[node style bb]{1}; \\
\node[node style bb]{1};&\node[node style bb]{0};&\node[node style bb]{0};&\node[node style bb]{1};&\node[node style bb]{1};&\node[node style bb]{1}; \\
\node[node style bb]{0};&\node[node style bb]{0};&\node[node style bb]{0};&0&1&1 \\
\node[node style bb]{0};&\node[node style bb]{0};&\node[node style bb]{0};&1&0&1 \\
  } ;
  
  \draw (m-1-4.north west) rectangle (m-4-6.south east);

\end{tikzpicture}
\raisebox{2.7em}{$\rightarrow$} 
\begin{tikzpicture}
  \matrix (m)[
    matrix of math nodes,
    left delimiter  = (,
    right delimiter = ),
  ] {
\node[node style bb]{0};&\node[node style bb]{1};&\node[node style bb]{1};&\node[node style bb]{1};&\node[node style bb]{1};&\node[node style bb]{1}; \\
\node[node style bb]{1};&\node[node style bb]{0};&\node[node style bb]{0};&\node[node style ba]{0};&\node[node style ba]{0};&\node[node style ba]{0}; \\
\node[node style bb]{0};&\node[node style bb]{0};&\node[node style bb]{0};&0&1&1 \\
\node[node style bb]{0};&\node[node style bb]{0};&\node[node style bb]{0};&1&0&1 \\
  } ;
  
  \draw (m-1-4.north west) rectangle (m-4-6.south east);

\end{tikzpicture}
\]
\end{minipage}
\begin{minipage}{0.6\textwidth}\footnotesize
  \Qcircuit @C=.8em @R=1em {
  	\lstick{x_1}& \qw & \targ & \targ & \qw & \rstick{\!\!\!\!x_1\oplus x_2\oplus x_4} \\
  	\lstick{x_2}& \targ & \qw & \ctrl{-1} & \qw & \rstick{\!\!\!\!x_2\oplus x_3} \\
  	\lstick{x_3}& \ctrl{-1} & \qw & \qw & \qw & \rstick{\!\!\!\!x_3} \\
  	\lstick{x_4}& \qw & \ctrl{-3} & \qw & \qw & \rstick{\!\!\!\!x_4}
  }
\end{minipage}

%%%%%%%%%%%%%%%%%%%%%%%%%%%%%%

Continuing on we recurse on the $0$-cofactor of row $3$ and apply a $\cnot$ with target bit $1$, control $4$, before backtracking to the $1$-cofactor.

\hspace*{-8em}
 \begin{minipage}{0.8\textwidth}
\[
\begin{tikzpicture}
  \matrix (m)[
    matrix of math nodes,
    left delimiter  = (,
    right delimiter = ),
  ] {
\node[node style bb]{0};&\node[node style bb]{1};&\node[node style bb]{1};&\node[node style bb]{1};&\node[node style bb]{1};&\node[node style bb]{1}; \\
\node[node style bb]{1};&\node[node style bb]{0};&\node[node style bb]{0};&\node[node style bb]{0};&\node[node style bb]{0};&\node[node style bb]{0}; \\
\node[node style bb]{0};&\node[node style bb]{0};&\node[node style bb]{0};&\node[node style bb]{0};&\node[node style bb]{1};&\node[node style bb]{1}; \\
\node[node style bb]{0};&\node[node style bb]{0};&\node[node style bb]{0};&1&0&1 \\
  } ;
  
  \draw (m-1-4.north west) rectangle (m-4-4.south east);

\end{tikzpicture}
\]
\end{minipage}
\begin{minipage}{0.6\textwidth}\footnotesize
  \Qcircuit @C=.8em @R=1em {
  	\lstick{x_1}& \qw & \targ & \targ & \qw & \rstick{\!\!\!\!x_1\oplus x_2\oplus x_4} \\
  	\lstick{x_2}& \targ & \qw & \ctrl{-1} & \qw & \rstick{\!\!\!\!x_2\oplus x_3} \\
  	\lstick{x_3}& \ctrl{-1} & \qw & \qw & \qw & \rstick{\!\!\!\!x_3} \\
  	\lstick{x_4}& \qw & \ctrl{-3} & \qw & \qw & \rstick{\!\!\!\!x_4}
  }
\end{minipage}

%%%%%%%%%%%%%%%%%%%%%%%%%%%%%%

\hspace*{-8em}
 \begin{minipage}{0.8\textwidth}
\[
\begin{tikzpicture}
  \matrix (m)[
    matrix of math nodes,
    left delimiter  = (,
    right delimiter = ),
  ] {
\node[node style bb]{0};&\node[node style bb]{1};&\node[node style bb]{1};&\node[node style bb]{1};&\node[node style bb]{1};&\node[node style bb]{1}; \\
\node[node style bb]{1};&\node[node style bb]{0};&\node[node style bb]{0};&\node[node style bb]{0};&\node[node style bb]{0};&\node[node style bb]{0}; \\
\node[node style bb]{0};&\node[node style bb]{0};&\node[node style bb]{0};&\node[node style bb]{0};&\node[node style bb]{1};&\node[node style bb]{1}; \\
\node[node style bb]{0};&\node[node style bb]{0};&\node[node style bb]{0};&\node[node style bb]{1};&0&1 \\
  } ;
  
  \draw (m-1-4.north west) rectangle (m-4-4.south east);

\end{tikzpicture}
\raisebox{2.7em}{$\rightarrow$} 
\begin{tikzpicture}
  \matrix (m)[
    matrix of math nodes,
    left delimiter  = (,
    right delimiter = ),
  ] {
\node[node style bb]{0};&\node[node style bb]{1};&\node[node style bb]{1};&\node[node style bb]{1};&\node[node style bb]{1};&\node[node style bb]{1}; \\
\node[node style bb]{1};&\node[node style bb]{0};&\node[node style bb]{0};&\node[node style bb]{0};&\node[node style bb]{0};&\node[node style bb]{0}; \\
\node[node style bb]{0};&\node[node style bb]{0};&\node[node style bb]{0};&\node[node style bb]{0};&\node[node style bb]{1};&\node[node style bb]{1}; \\
\node[node style bb]{0};&\node[node style bb]{0};&\node[node style bb]{0};&\node[node style ba]{0};&\node[node style aa]{1};&\node[node style aa]{0}; \\
  } ;
  
  \draw (m-1-4.north west) rectangle (m-4-4.south east);

\end{tikzpicture}
\]
\end{minipage}
\begin{minipage}{0.6\textwidth}\footnotesize
  \Qcircuit @C=.8em @R=1em {
  	\lstick{x_1}& \qw & \targ & \targ & \targ & \qw & \rstick{\!\!\!\!x_1\oplus x_2} \\
  	\lstick{x_2}& \targ & \qw & \ctrl{-1} & \qw & \qw & \rstick{\!\!\!\!x_2\oplus x_3} \\
  	\lstick{x_3}& \ctrl{-1} & \qw & \qw & \qw & \qw & \rstick{\!\!\!\!x_3} \\
  	\lstick{x_4}& \qw & \ctrl{-3} & \qw & \ctrl{-3} & \qw & \rstick{\!\!\!\!x_4}
  }
\end{minipage}

%%%%%%%%%%%%%%%%%%%%%%%%%%%%%%

For the remaining two columns we first zero out row $3$ by applying a $\cnot$ gate between bits $3$ and $1$, then finally descend into the cofactors on the last row.

\hspace*{-8em}
 \begin{minipage}{0.8\textwidth}
\[
\begin{tikzpicture}
  \matrix (m)[
    matrix of math nodes,
    left delimiter  = (,
    right delimiter = ),
  ] {
\node[node style bb]{0};&\node[node style bb]{1};&\node[node style bb]{1};&\node[node style bb]{1};&\node[node style bb]{1};&\node[node style bb]{1}; \\
\node[node style bb]{1};&\node[node style bb]{0};&\node[node style bb]{0};&\node[node style bb]{0};&\node[node style bb]{0};&\node[node style bb]{0}; \\
\node[node style bb]{0};&\node[node style bb]{0};&\node[node style bb]{0};&\node[node style bb]{0};&\node[node style bb]{1};&\node[node style bb]{1}; \\
\node[node style bb]{0};&\node[node style bb]{0};&\node[node style bb]{0};&\node[node style bb]{0};&1&0 \\
  } ;

  \draw (m-1-5.north west) rectangle (m-4-6.south east);

\end{tikzpicture}
\raisebox{2.7em}{$\rightarrow$} 
\begin{tikzpicture}
  \matrix (m)[
    matrix of math nodes,
    left delimiter  = (,
    right delimiter = ),
  ] {
\node[node style bb]{0};&\node[node style bb]{1};&\node[node style bb]{1};&\node[node style bb]{1};&\node[node style bb]{1};&\node[node style bb]{1}; \\
\node[node style bb]{1};&\node[node style bb]{0};&\node[node style bb]{0};&\node[node style bb]{0};&\node[node style bb]{0};&\node[node style bb]{0}; \\
\node[node style bb]{0};&\node[node style bb]{0};&\node[node style bb]{0};&\node[node style bb]{0};&\node[node style ba]{0};&\node[node style ba]{0}; \\
\node[node style bb]{0};&\node[node style bb]{0};&\node[node style bb]{0};&\node[node style bb]{0};&1&0 \\
  } ;

  \draw (m-1-5.north west) rectangle (m-4-6.south east);

\end{tikzpicture}
\]
\end{minipage}
\begin{minipage}{0.6\textwidth}\footnotesize
  \Qcircuit @C=.8em @R=1em {
  	\lstick{x_1}& \qw & \targ & \targ & \targ & \targ & \qw & \rstick{\!\!\!\!x_1\oplus x_2\oplus x_3} \\
  	\lstick{x_2}& \targ & \qw & \ctrl{-1} & \qw & \qw & \qw & \rstick{\!\!\!\!x_2\oplus x_3} \\
  	\lstick{x_3}& \ctrl{-1} & \qw & \qw & \qw & \ctrl{-2} & \qw & \rstick{\!\!\!\!x_3} \\
  	\lstick{x_4}& \qw & \ctrl{-3} & \qw & \ctrl{-3} & \qw & \qw & \rstick{\!\!\!\!x_4}
  }
\end{minipage}

%%%%%%%%%%%%%%%%%%%%%%%%%%%%%%

\hspace*{-8em}
 \begin{minipage}{0.8\textwidth}
\[
\begin{tikzpicture}
  \matrix (m)[
    matrix of math nodes,
    left delimiter  = (,
    right delimiter = ),
  ] {
\node[node style bb]{0};&\node[node style bb]{1};&\node[node style bb]{1};&\node[node style bb]{1};&\node[node style bb]{1};&\node[node style bb]{1}; \\
\node[node style bb]{1};&\node[node style bb]{0};&\node[node style bb]{0};&\node[node style bb]{0};&\node[node style bb]{0};&\node[node style bb]{0}; \\
\node[node style bb]{0};&\node[node style bb]{0};&\node[node style bb]{0};&\node[node style bb]{0};&\node[node style bb]{0};&\node[node style bb]{0}; \\
\node[node style bb]{0};&\node[node style bb]{0};&\node[node style bb]{0};&\node[node style bb]{0};&\node[node style bb]{1};&\node[node style bb]{0}; \\
  } ;

  \draw (m-1-6.north west) rectangle (m-4-6.south east);

\end{tikzpicture}
\]
\end{minipage}
\begin{minipage}{0.6\textwidth}\footnotesize
  \Qcircuit @C=.8em @R=1em {
  	\lstick{x_1}& \qw & \targ & \targ & \targ & \targ & \qw & \rstick{\!\!\!\!x_1\oplus x_2\oplus x_3} \\
  	\lstick{x_2}& \targ & \qw & \ctrl{-1} & \qw & \qw & \qw & \rstick{\!\!\!\!x_2\oplus x_3} \\
  	\lstick{x_3}& \ctrl{-1} & \qw & \qw & \qw & \ctrl{-2} & \qw & \rstick{\!\!\!\!x_3} \\
  	\lstick{x_4}& \qw & \ctrl{-3} & \qw & \ctrl{-3} & \qw & \qw & \rstick{\!\!\!\!x_4}
  }
\end{minipage}

%%%%%%%%%%%%%%%%%%%%%%%%%%%%%%

\hspace*{-8em}
 \begin{minipage}{0.8\textwidth}
\[
\begin{tikzpicture}
  \matrix (m)[
    matrix of math nodes,
    left delimiter  = (,
    right delimiter = ),
  ] {
\node[node style bb]{0};&\node[node style bb]{1};&\node[node style bb]{1};&\node[node style bb]{1};&\node[node style bb]{1};&\node[node style bb]{1}; \\
\node[node style bb]{1};&\node[node style bb]{0};&\node[node style bb]{0};&\node[node style bb]{0};&\node[node style bb]{0};&\node[node style bb]{0}; \\
\node[node style bb]{0};&\node[node style bb]{0};&\node[node style bb]{0};&\node[node style bb]{0};&\node[node style bb]{0};&\node[node style bb]{0}; \\
\node[node style bb]{0};&\node[node style bb]{0};&\node[node style bb]{0};&\node[node style bb]{0};&\node[node style bb]{1};&\node[node style bb]{0}; \\
  } ;

  \draw (m-1-5.north west) rectangle (m-4-5.south east);

\end{tikzpicture}
\raisebox{2.7em}{$\rightarrow$} 
\begin{tikzpicture}
  \matrix (m)[
    matrix of math nodes,
    left delimiter  = (,
    right delimiter = ),
  ] {
\node[node style bb]{0};&\node[node style bb]{1};&\node[node style bb]{1};&\node[node style bb]{1};&\node[node style bb]{1};&\node[node style bb]{1}; \\
\node[node style bb]{1};&\node[node style bb]{0};&\node[node style bb]{0};&\node[node style bb]{0};&\node[node style bb]{0};&\node[node style bb]{0}; \\
\node[node style bb]{0};&\node[node style bb]{0};&\node[node style bb]{0};&\node[node style bb]{0};&\node[node style bb]{0};&\node[node style bb]{0}; \\
\node[node style bb]{0};&\node[node style bb]{0};&\node[node style bb]{0};&\node[node style bb]{0};&\node[node style ba]{0};&\node[node style bb]{0}; \\
  } ;

  \draw (m-1-5.north west) rectangle (m-4-5.south east);

\end{tikzpicture}
\]
\end{minipage}
\begin{minipage}{0.6\textwidth}\footnotesize
  \Qcircuit @C=.8em @R=1em {
  	\lstick{x_1}& \qw & \targ & \targ & \targ & \targ & \targ & \qw & \rstick{\!\!\!\!x_1\oplus x_2\oplus x_3\oplus x_4} \\
  	\lstick{x_2}& \targ & \qw & \ctrl{-1} & \qw & \qw & \qw & \qw & \rstick{\!\!\!\!x_2\oplus x_3} \\
  	\lstick{x_3}& \ctrl{-1} & \qw & \qw & \qw & \ctrl{-2} & \qw & \qw & \rstick{\!\!\!\!x_3} \\
  	\lstick{x_4}& \qw & \ctrl{-3} & \qw & \ctrl{-3} & \qw & \ctrl{-3} & \qw & \rstick{\!\!\!\!x_4}
  }
\end{minipage}

The overall linear transformation applied is
\[
  A = \begin{bmatrix} 1 & 1 & 0 & 1 \\ 0 & 1 & 1 & 0 \\ 0 & 0 & 1 & 0 \\ 0 & 0 & 0 & 1 \end{bmatrix}
\]
so the algorithm completes by appending a circuit computing $A^{-1}$. Inserting $T=\rz(1/8)$ gates in the relevant positions, we get the following circuit computing $\ket{\x}\mapsto e^{2\pi i f(\x)}\ket{\x}$:
{\footnotesize \[
{
  \Qcircuit @C=.8em @R=1em {
  	\lstick{x_1}& \qw & \qw & \gate{T} & \targ & \gate{T} & \targ & \targ & \gate{T} & \targ & \gate{T} & \targ & \gate{T} & \targ & \qw & \targ & \qw & \rstick{\!\!\!\!x_1} \\
  	\lstick{x_2}& \targ & \gate{T} & \qw & \qw & \qw & \ctrl{-1} & \qw & \qw & \qw & \qw & \qw & \qw & \ctrl{-1} & \targ & \qw & \qw & \rstick{\!\!\!\!x_2} \\
  	\lstick{x_3}& \ctrl{-1} & \qw & \qw & \qw & \qw & \qw & \qw & \qw & \ctrl{-2} & \qw & \qw & \qw & \qw & \ctrl{-1} & \qw & \qw & \rstick{\!\!\!\!x_3} \\
  	\lstick{x_4}& \qw & \qw & \qw & \ctrl{-3} & \qw & \qw & \ctrl{-3} & \qw & \qw & \qw & \ctrl{-3} & \qw & \qw & \qw & \ctrl{-3} & \qw & \rstick{\!\!\!\!x_4}
  }
%  \Qcircuit @C=.8em @R=1em {
%  	\lstick{x_1}& \qw & \qw & \qw & \ctrl{3} & \qw & \qw & \ctrl{3} & \qw & \qw & \qw & \ctrl{3} & \qw & 
%  		\qw & \qw & \ctrl{3} & \qw & \rstick{x_1} \\
%  	\lstick{x_2}& \ctrl{1} & \qw & \qw & \qw & \qw & \qw & \qw & \qw & \ctrl{2} & \qw & \qw & \qw & 
%  		\qw & \ctrl{1} & \qw & \qw & \rstick{x_2} \\
%  	\lstick{x_3}& \targ & \gate{T} & \qw & \qw & \qw & \ctrl{1} & \qw & \qw & \qw & \qw & \qw & \qw & 
%  		\ctrl{1} & \targ & \qw & \qw & \rstick{x_3} \\
%  	\lstick{x_4}& \qw & \qw & \gate{T} & \targ  & \gate{T} & \targ & \targ & \gate{T} & \targ & \gate{T} & \targ & \gate{T} & 
%  		\targ & \qw & \targ & \qw & \rstick{x_4}
%  }
  }
\]}

\end{example}

\begin{example}
\Cref{fig:ccz} shows a circuit implementing the doubly-controlled $Z$ gate, corresponding to an identity parity network for $\F^3\setminus\{000\}$, synthesized with \cref{alg:graysynth} followed by the Patel-Markov-Hayes algorithm. In this case both the parity network and the identity parity network are minimal, as verified by brute force search -- further, the circuit synthesized by \cref{alg:graysynth} reproduces exactly the minimal circuit for $\F^3\setminus\{000\}$ from \cite{wgma14}. In general, for any $n$ the identity parity network for $\F^n\setminus\{\0\}$ synthesized in this manner has the same structure, using $2^n-2$ $\cnot$ gates, compared to $2^n$ bit flips for the Gray code.
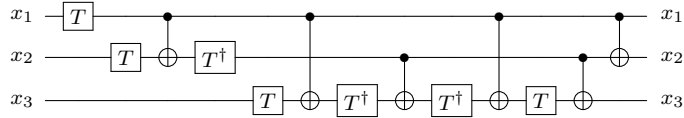
\begin{figure}
{\footnotesize\[
{
  \Qcircuit @C=.8em @R=.5em {
  	\lstick{x_1} & \gate{T} & \qw & \ctrl{1} & \qw & \qw & \ctrl{2} & \qw & \qw & \qw & \ctrl{2} & \qw& \qw & \ctrl{1} & \qw & \rstick{\!\!\!\!x_1} \\
  	\lstick{x_2} & \qw & \gate{T} & \targ & \gate{T^\dagger} & \qw & \qw & \qw & \ctrl{1} & \qw & \qw & \qw & \ctrl{1} & \targ & \qw & \rstick{\!\!\!\!x_2} \\
  	\lstick{x_3} & \qw & \qw & \qw & \qw & \gate{T} & \targ & \gate{T^\dagger} & \targ & \gate{T^\dagger} & \targ & \gate{T} & \targ & \qw & \qw & \rstick{\!\!\!\!x_3}
  }}
\]}
\caption{Circuit implementing the doubly-controlled $Z$ gate $\Lambda_2(Z)$ synthesized with \cref{alg:graysynth}. The $\cnot$-minimal Fourier expansion in this case gives $S=\F^3\setminus\{000\}$, and $A=I$.}
\label{fig:ccz}
\end{figure}
\end{example}

\begin{example}
If instead of $S=\F^n\setminus\{\0\}$ we have $S\iso \F^m$ for some $m<n$, \cref{alg:graysynth} instead gives a circuit corresponding directly to the Gray code. In particular, \cref{fig:grayskel} shows a parity network for $S = \{ (\y \concat 1) \mid \y\in\F^3\}\iso \F^3$ synthesized with \cref{alg:graysynth}, where $S\iso \F^3$. In this case it can be observed that the controls of the $\cnot$ gates are exactly the bits flipped in a Gray code for $\F^3$. Further, it may be noted that this is a minimal size parity network for $S$, and is in fact a fixed-target parity network.

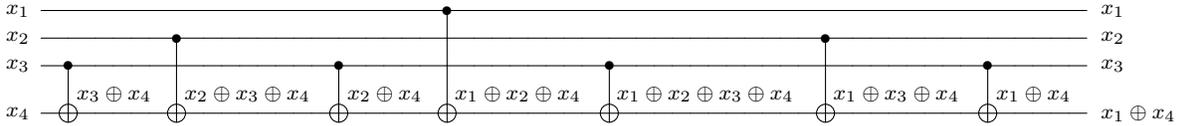
\begin{figure}
{\footnotesize\[\hspace*{-2.5em}
{
  \Qcircuit @C=.8em @R=1em {
  	\lstick{x_1} & \qw & \qw& \qw& \qw& \qw& \qw& \qw& \qw& \qw& \qw& \qw& \qw& \qw& \qw & \qw & \qw & \qw & \qw & \ctrl{3} & \qw & \qw& \qw& \qw& \qw& \qw& \qw& \qw& \qw& \qw& \qw& \qw& \qw& \qw& \qw& \qw& \qw& \qw& \qw& \qw& \qw& \qw& \qw& \qw& \qw& \qw & \qw & \qw & \qw & \qw & \qw & \qw & \rstick{\!\!\!\!x_1} \\
  	\lstick{x_2} & \qw & \qw& \qw& \qw& \qw & \ctrl{2} & \qw & \qw& \qw& \qw& \qw& \qw& \qw& \qw& \qw& \qw& \qw& \qw& \qw& \qw& \qw& \qw& \qw& \qw& \qw& \qw& \qw& \qw& \qw& \qw& \qw& \qw & \qw & \qw & \qw & \qw & \qw & \ctrl{2} & \qw& \qw& \qw& \qw& \qw& \qw& \qw& \qw& \qw& \qw & \qw & \qw & \qw & \rstick{\!\!\!\!x_2}  \\
  	\lstick{x_3} & \ctrl{1} & \qw& \qw& \qw& \qw& \qw& \qw& \qw& \qw& \qw& \qw & \qw & \qw & \ctrl{1} & \qw& \qw& \qw& \qw& \qw& \qw& \qw& \qw& \qw& \qw & \qw & \qw & \ctrl{1} & \qw& \qw& \qw& \qw& \qw& \qw& \qw& \qw& \qw& \qw& \qw& \qw& \qw& \qw& \qw& \qw & \qw & \qw & \ctrl{1} & \qw& \qw& \qw& \qw & \qw & \rstick{\!\!\!\!x_3}  \\
  	\lstick{x_4} & \targ & \ustick{\;\;\;\;\;\;x_3\oplus x_4}\qw & \push{\rule{0em}{2em}}\qw& \qw& \qw& \targ & \ustick{\;\;\;\;\;\;\;\;\;\;\;\;\;\;x_2\oplus x_3\oplus x_4}\qw & \qw& \qw& \qw& \qw& \qw& \qw& \targ & \ustick{\;\;\;\;\;\;x_2\oplus x_4}\qw & \qw& \qw& \qw& \targ & \ustick{\;\;\;\;\;\;\;\;\;\;\;\;\;\;x_1\oplus x_2\oplus x_4}\qw & \qw& \qw& \qw& \qw& \qw& \qw& \targ & \ustick{\;\;\;\;\;\;\;\;\;\;\;\;\;\;\;\;\;\;\;\;\;\;x_1\oplus x_2\oplus x_3\oplus x_4}\qw & \qw& \qw& \qw& \qw& \qw& \qw& \qw& \qw& \qw& \targ & \ustick{\;\;\;\;\;\;\;\;\;\;\;\;\;\;x_1\oplus x_3\oplus x_4}\qw & \qw& \qw& \qw& \qw& \qw& \qw& \targ & \ustick{\;\;\;\;\;\;x_1\oplus x_4}\qw & \qw& \qw& \qw& \qw & \rstick{\!\!\!\!x_1\oplus x_4}
  }}
\]}
\caption{Annotated parity network for the set $S = \{ (\y \concat 1) \mid \y\in\F^3\}$. Note that the parity network corresponds exactly to the Gray code on $\F^3$.}
\label{fig:grayskel}
\end{figure}
\end{example}

\subsection{Synthesis with encoded inputs}\label{sec:encoded}

Given an encoder $E\in\F^{m\times n}$, we can use \cref{alg:graysynth} to synthesize a parity network for some set $S$ with encoded inputs as follows. Recall that a parity network for $S\subseteq \F^n$ with inputs encoded by $E$ corresponds to a parity network for some set $S'\subseteq \F^m$ such that for any $\y\in S$, there exists $\w\in S'$ where $E^T\w = \y$. While finding the minimal such $\w$ would require solving the NP-hard Maximum-likelihood decoding problem, we can efficiently compute \emph{some} $\w$ using a \emph{generalized inverse}.

Recall that a generalized inverse $A^g$ of a matrix $A\in\F^{m\times n}$ is an $n$ by $m$ matrix over $\F$ such that
\[
  AA^gA = A,
\]
and in particular $AA^g\y = \y$ whenever there exists $\x$ such that $A\x = \y$. We use such a generalized inverse rather than a specific inverse such as the Moore-Penrose pseudoinverse, as the latter typically does not exist over finite fields. By contrast, for $A\in \F^{m\times n}$, a generalized inverse $A^g$ may be computed \cite{cm09} by finding invertible matrices $P, Q$ over $\F$ such that
\[
  A = P\begin{bmatrix} I_r & 0 \\ 0 & 0 \end{bmatrix} Q
\]
and computing 
\[
  A^g = Q^{-1}\begin{bmatrix} I_r & 0 \\ 0 & 0 \end{bmatrix} P^{-1}.
\]
$P$ and $Q$ may likewise be found by first reducing $A$ to row-echelon form, then reducing its transpose to row-echelon form.

It is worth noting that it may be possible to perform additional optimization by optimizing the set $S'$ with $A^g$. In particular, it is known \cite{cm09} that the set of solutions to the linear system $A\x = \y$ is given by $\{ A^g\y + (I - A^gA)\w \mid \w\in\F^m\}$, which may be possible to optimize with classical techniques. We tried brute-force optimizing the set $S'$ for some small instances and found negligible effects on overall $\cnot$ counts, though we leave it as an open question as to whether scalable sub-optimal methods reduce parity network sizes in large benchmarks.

%%%%%%%%%%%%%%%%%%%%%%%%%%%%%%%%%%%%%%%%
\section{Evaluation}

We implemented \cref{alg:graysynth} in Haskell in the open-source quantum circuit toolkit {\sc Feynman}\footnote{\href{https://github.com/meamy/feynman/tree/graysynth}{https://github.com/meamy/feynman}}. Experiments were run in Debian Linux running on a quad-core 64-bit Intel Core i7 2.40 GHz processor and 8 GB RAM. 

\begin{figure}
\includegraphics[scale=0.8]{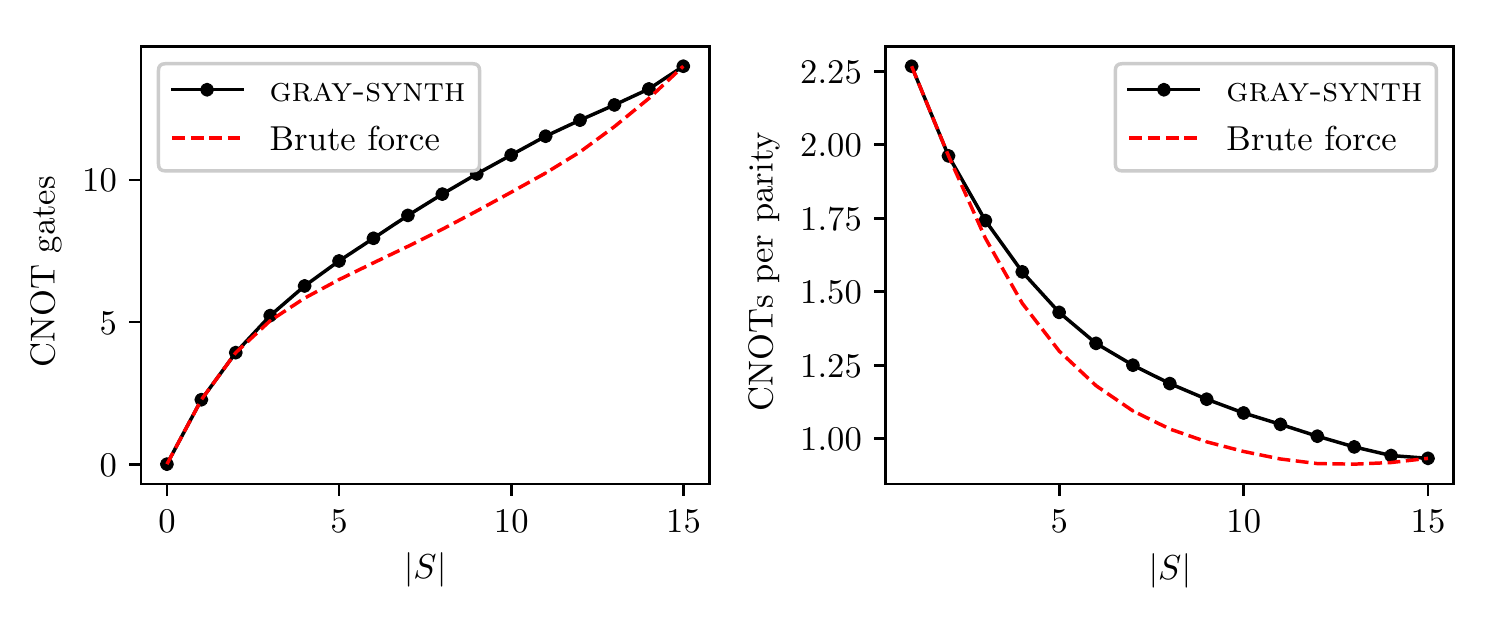}
\caption{Average $\cnot$ counts of parity networks as computed by \cref{alg:graysynth} and brute force minimization over all sets of $4$ bit parities.}
\label{fig:brute}
\end{figure}

We generated all $4$-bit minimal parity networks by brute force search and compared them with the parity networks generated with \cref{alg:graysynth}. \Cref{fig:brute} graphs the results, with $\cnot$ cost averaged over sets $S$ of parities with the same size. The results show that our algorithm synthesizes optimal or near-optimal networks for small and large sets $S$, and diverges slightly for sets of parities containing around half of the possible parities. The divergence peaks at $|S|=8$, exactly half of the $2^4$ parities, with \cref{alg:graysynth} coming within 15\% of the minimal number of $\cnot$ gates on average. On examining the structure of optimal parity networks for sets on which \cref{alg:graysynth} performed poorly, it appears that the optimal parity networks save on $\cnot$ cost by making more judicious use of \emph{shortcuts} -- leaving qubits in particular states to flip between distant parities in other bits quickly. In general we found that the optimal results in these cases \emph{were not} achievable just by using the {\sc gray-synth} algorithm with different index expansion orders. An effective synthesis algorithm may then be to combine \cref{alg:graysynth} for small and large sets with a different heuristic for sets $S$ of size close to $2^{n-1}$.

\subsection{Benchmarks}

To evaluate the performance of \cref{alg:graysynth} on practical quantum circuits, we implemented a variant of the $T$-par algorithm \cite{amm14} with \cref{alg:graysynth} replacing the original Matroid partitioning-based $\{\cnot, \rz\}$ synthesis, which optimized $T$-depth rather than $\cnot$-count. Circuits written over the Clifford+$T$ gate set are analyzed in sum-over-paths form and effectively split into alternating sequences of $\{\cnot, \rz\}$ gates and Hadamard gates, the former of which are synthesized with \cref{alg:graysynth}. As the input basis state of a particular $\{\cnot, \rz\}$ sub-circuit is expressed by an encoding of the entire circuit's inputs as well as \emph{path variables} \cite{amm14, dhmhno05}, we use the method of \cref{sec:encoded} to synthesize the circuit. As \cref{alg:graysynth} is only used to perform the $\{\cnot, \rz\}$ synthesis sub-routine, our optimized circuits have the same $T$-count as the original $T$-par algorithm.

We ran our implementation on a suite of benchmarks previously used to evaluate quantum circuit optimizations \cite{amm14, am16}. Benchmark circuits containing Toffoli gates and multiply controlled Toffoli gates were expanded to the Clifford+$T$ gate set using the Toffoli gate decomposition from \cite{ammr13} with $7$ $\cnot$ gates, and the Nielsen-Chuang multiply controlled Toffoli decomposition \cite{nc00} giving $2(k-2)$ Toffoli gates per $k$-controlled Toffoli. While other decompositions, particularly with better $\cnot$ counts exist, we chose these standard decompositions to compare against previous results, which used the same decompositions. All but 4 optimized circuits\footnote{It was reported in \cite{a18} that 4 separate optimized benchmark circuits (CSLA-MUX\_3, Adder\_8, Mod-Mult\_55 and GF($2^{32}$)) provably contained errors. These errors have since been fixed and these circuits now pass verification.} (Mod-Adder$\_{1024}$, Cycle $17\_3$, GF($2^{64}$)-Mult and HWB$\_8$) were formally verified to be correct using the method of \cite{a18}, with verification failing on the remaining $4$ due to lack of memory. To the best of the authors' knowledge, this also represents the first quantum circuit optimization work where a majority of the optimized circuits have also been formally verified.

\afterpage{%
\clearpage% Flush earlier floats (otherwise order might not be correct)
\thispagestyle{empty}% empty page style (?)
\begin{landscape}
\begin{table}
\footnotesize
\hspace*{-1em}
\begin{tabular}{lr>{\bf}rrrr>{\bf}rrrr>{\bf}rrr>{\bf}rrrr} \toprule
Benchmark & $n$ & \multicolumn{3}{c}{Base} & \multicolumn{4}{c}{$T$-par (matroids)} 
	& \multicolumn{3}{c}{Nam \etal (L)} & \multicolumn{5}{c}{$T$-par ({\sc gray-synth})} \\ 
\cmidrule(l{4pt}r{1pt}){3-5} \cmidrule(l{4pt}r{1pt}){6-9} \cmidrule(l{4pt}r{1pt}){10-12} \cmidrule(l{4pt}r{1pt}){13-17}
 & & $\cnot$ & $T$ & $T$-depth & Time (s) & $\cnot$ & $T$ & $T$-depth & Time & $\cnot$ & $T$ & Time & $\cnot$ & $T$ & $T$-depth & \% Red. \\
Grover$\_5$ & 9 & 336 & 336 & 144 & $0.028$ & 425 & 154 & 59 & -- & -- & -- & $0.001$ & 226 & 154 & 128 & 32.7  \\ %\hline
Mod $5\_4$ & 5 & 32 & 28 & 12 & $0.001$ & 50 & 16 & 7 & $<0.001$ & 28 & 16 & $0.001$ & 26 & 16 & 15 & 18.8  \\ %\hline
VBE-Adder$\_3$ & 10 & 80 & 70 & 24 & $0.004$ & 91 & 24 & 9 & $<0.001$ & 50 & 24 & $0.004$ & 46 & 24 & 22 & 42.5  \\ %\hline
CSLA-MUX$\_3$ & 15 & 90 & 70 & 21 & $0.011$ & 213 & 62 & 13 & $<0.001$ & 76 & 64 & $0.073$ & 118 & 62 & 29 & -31.1  \\
CSUM-MUX$\_9$ & 30 & 196 & 196 & 18 & $0.038$ & 268 & 84 & 16 & $<0.001$ & 168 & 84 & $0.095$ & 148 & 84 & 23 & 24.5  \\ %\hline
QCLA-Com$\_7$ & 24 & 215 & 203 & 27 & $0.055$ & 271 & 94 & 20 & $0.001$ & 132 & 95 & $0.097$ & 136 & 94 & 28 & 36.7  \\
QCLA-Mod$\_7$ & 26 & 441 & 413 & 66 & $0.099$ & 729 & 237 & 58 & $0.004$ & 302 & 237 & $0.145$ & 360 & 237 & 67 & 18.4  \\
QCLA-Adder$\_10$ & 36 & 267 & 238 & 24 & $0.076$ & 442 & 162 & 19 & $0.002$ & 195 & 162 & $0.112$ & 214 & 162 & 30 & 19.9  \\ %\hline
Adder$\_{8}$ & 24 & 466 & 399 & 69 & $0.082$ & 654 & 215 & 54 & $0.004$ & 331 & 215 & $0.165$ & 359 & 215 & 73 & 23.0  \\% \hline
RC-Adder$\_{6}$ & 14 & 104 & 77 & 33 & $0.012$ & 133 & 47 & 14 & $<0.001$ & 73 & 47 & $0.080$ & 71 & 47 & 36 & 31.7  \\ %\hline
Mod-Red$\_{21}$ & 11 & 122 & 119 & 48 & $0.076$ & 191 & 73 & 28 & $<0.001$ & 81 & 73 & $0.091$ & 86 & 73 & 59 & 29.5  \\
Mod-Mult$\_{55}$ & 9 & 55 & 49 & 15 & $0.003$ & 89 & 35 & 9 & $<0.001$ & 40 & 35 & $0.004$ & 40 & 35 & 20 & 27.3  \\ %\hline
Mod-Adder$\_{1024}$ & 28 & 2005 & 1995 & 831 & $0.608$ & 2842 & 1011 & 249 & -- & -- & -- & $0.739$ & 1390 & 1011 & 863 & 30.7  \\ %\hline
Mod-Adder$\_{1048576}$\!\!\!\! & 58 & 16680 & 16660 & 7292 & 39.565 & 26600 & 7339 & 5426 & -- & -- & -- & 12.272 & 11080 & 7339 & 6570 & 33.6  \\ %\hline
Cycle $17\_3$ & 35 & 4532 & 4739 & 2001 & $2.431$ & 6547 & 1955 & 315 & -- & -- & -- & $2.618$ & 2968 & 1955 & 1857 & 37.4  \\ %\hline
GF($2^4$)-Mult & 12 & 115 & 112 & 36 & $0.006$ & 197 & 68 & 14 & $0.001$ & 99 & 68 & $0.041$ & 106 & 68 & 39 & 7.8  \\
GF($2^5$)-Mult & 15 & 179 & 175 & 51 & $0.012$ & 334 & 111 & 17 & $0.001$ & 154 & 115 & $0.038$ & 163 & 111 & 53 & 8.9  \\
GF($2^6$)-Mult & 18 & 257 & 252 & 60 & $0.019$ & 484 & 150 & 22 & $0.003$ & 221 & 150 & $0.055$ & 235 & 150 & 63 & 8.6  \\
GF($2^7$)-Mult & 21 & 349 & 343 & 72 & $0.044$ & 713 & 217 & 27 & $0.004$ & 300 & 217 & $0.450$ & 319 & 217 & 75 & 8.6  \\
GF($2^8$)-Mult & 24 & 469 & 448 & 84 & $0.047$ & 932 & 264 & 31 & $0.006$ & 405 & 264 & $0.066$ & 428 & 264 & 87 & 8.7  \\
GF($2^9$)-Mult & 27 & 575 & 567 & 96 & $0.113$ & 1198 & 351 & 34 & $0.010$ & 494 & 351 & $0.076$ & 526 & 351 & 95 & 8.5  \\
GF($2^{10}$)-Mult & 30 & 709 & 700 & 108 & $0.147$ & 1462 & 410 & 36 & $0.009$ & 609 & 410 & $0.081$ & 648 & 410 & 109 & 8.6  \\
GF($2^{16}$)-Mult & 48 & 1837 & 1792 & 180 & $1.519$ & 4556 & 1040 & 65 & $0.065$ & 1581 & 1040 & $0.363$ & 1691 & 1040 & 585 & 7.9  \\ 
GF($2^{32}$)-Mult & 96 & 7292 & 7168 & 372 & $132.133$ & 22205 & 4128 & 126 & $1.834$ & 6299 & 4128 & $5.571$ & 6636 & 4128 & 2190 & 9.0   \\ 
GF($2^{64}$)-Mult & 192 & 28861 & 28672 & 756 & $19072.290$ & 105830 & 16448 & 256 & 58.341 & 24765 & 16448 & $114.310$  & 25934 & 16448 & 7716 & 10.1   \\ %\hline
Ham$\_{15}$ (low) & 17 & 259 & 161 & 69 & $0.069$ & 376 & 97 & 38 & -- & -- & -- & $0.043$ & 208 & 97 & 83 & 19.7  \\
Ham$\_{15}$ (med) & 17 & 616 & 574 & 240 & $0.108$ & 695 & 242 & 98 & -- & -- & -- & $0.089$ & 357 & 242 & 201 & 42.0  \\
Ham$\_{15}$ (high) & 20 & 2500 & 2457 & 996 & $0.498$ & 3036 & 1021 & 424 & -- & -- & -- & $0.376$ & 1502 & 1021 & 837 & 39.9  \\% \hline
HWB$\_6$ & 7 & 131 & 105 & 45 & $0.006$ & 199 & 75 & 25 & -- & -- & -- & $0.029$ & 110 & 75 & 63 & 16.0  \\
HWB$\_8$ & 12 & 7508 & 5425 & 2106 & $1.879$ & 12428 & 3531 & 951 & -- & -- & -- & $1.706$ & 6861 & 3531 & 2752 & 8.6  \\% \hline
QFT$\_4$ & 5 & 48 & 69 & 48 & $0.006$ & 76 & 67 & 45 & -- & -- & -- & $0.005$ & 48 & 67 & 62 & 0.0  \\ %\hline
$\Lambda_3(X)$ & 5 & 21 & 21 & 9 & $0.001$ & 32 & 15 & 8 & $<0.001$ & 14 & 15 & $<0.001$ & 14 & 15 & 12 & 33.3  \\
$\Lambda_3(X)$ (Barenco) & 5 & 28 & 28 & 12 & $0.015$ & 29 & 16 & 7 & $<0.001$ & 20 & 16 & $<0.001$ & 18 & 16 & 16 & 35.7  \\
$\Lambda_4(X)$ & 7 & 35 & 35 & 15 & $0.002$ & 55 & 23 & 7 & $<0.001$ & 22 & 23 & $0.001$ & 22 & 23 & 18 & 37.1  \\
$\Lambda_4(X)$ (Barenco) & 7 & 56 & 56 & 24 & $0.002$ & 71 & 28 & 11 & $<0.001$ & 40 & 28 & $<0.001$ & 36 & 28 & 26 & 35.7  \\
$\Lambda_5(X)$ & 9 & 49 & 49 & 21 & $0.004$ & 76 & 31 & 10 & $<0.001$ & 30 & 31 & $0.003$ & 30 & 31 & 24 & 38.8  \\
$\Lambda_5(X)$ (Barenco) & 9 & 84 & 84 & 36 & $0.004$ & 103 & 40 & 15 & $<0.001$ & 60 & 40 & $<0.001$ & 54 & 40 & 34 & 35.7  \\
$\Lambda_{10}(X)$ & 19 & 119 & 119 & 51 & $0.066$ & 167 & 71 & 21 & $<0.001$ & 70 & 71 & $0.071$ & 70 & 71 & 54 & 41.2  \\ 
$\Lambda_{10}(X)$ (Barenco) & 19 & 224 & 224 & 96 & $0.024$ & 271 & 100 & 25 & $0.001$ & 160 & 100 & $0.029$ & 144 & 100 & 81 & 35.7  \\ \midrule
\multicolumn{16}{l}{Total} & 22.6 \\ \bottomrule
\end{tabular}
\caption{Benchmark optimization results. Base gives the original circuit statistics, $T$-par (matroids) gives optimization results using \cite{amm14}, Nam \etal (L) reports the light optimization results from \cite{nrscm17} (where available), and $T$-par ({\sc gray-synth}) gives the results using \cref{alg:graysynth} instead of matroid partitioning. The \% reduction in $\cnot$ gates over the base circuit is reported in the last column.}\label{tab:results}
\end{table}
\end{landscape}
}

\Cref{tab:results} reports the results of our experiments. On average, \cref{alg:graysynth} resulted in an 23\% reduction of $\cnot$ gates, with 43\% reduction in the best case. In reality there may be more $\cnot$ reduction on average, as the algorithm performed relatively poorly on the Galois field multipliers, which comprise over a quarter of the benchmarks. Further, only 4 benchmarks took over a second to complete, lending evidence to the scalability of our method. One benchmark, CSLA-MUX$\_3$, did observe an \emph{increase} in $\cnot$ gates of $31\%$ -- this appears to be due to our sub-optimal method of generating a pointed parity network from a parity network, combined with the fact that very few $T$ gates cancel (that is, the support of the Fourier expansions synthesized have total size close to the number of $T$ gates in the original circuit). It may be possible to reduce the overhead in this case, and further reduce $\cnot$ counts in the other benchmarks, by synthesizing pointed parity network directly, rather than synthesizing a parity network followed by a linear permutation.

As \cref{alg:graysynth} effectively replaces the matroid partitioning algorithm used in $T$-par \cite{amm14}, optimizing $\cnot$ count rather than $T$-depth, we compared the time, $\cnot$ and $T$-depth tradeoffs of the two methods. The results are presented in \cref{tab:results}. While \cref{alg:graysynth} reduces the $\cnot$-count significantly compared to matroid partitioning, the $T$-depth is significantly increased in those cases, as expected. In fact, the $T$-depth and $\cnot$ count appear to be inversely proportional -- however, unlike matroid paritioning which increases the $\cnot$ count of all benchmarks, \cref{alg:graysynth} frequently reduces the $T$-depth of the original circuit. The runtimes are also significantly reduced in large circuits with a large number of $T$ gates, which is expected as the asymptotic complexity of the matroid partitioning algorithm used in $T$-par is significantly higher \cite{amm14}.

We also compared the results of our heuristic optimization algorithm to a recent heuristic by Nam, Ross, Su, Childs and Maslov \cite{nrscm17}, which was brought to our attention while preparing this manuscript. While their algorithm does not explicitly look at the problem of $\cnot$ minimization, they achieve significant $\cnot$ reductions by combining (among other techniques) a $T$-par style phase gate folding stage with optimized decompositions of specific gates and rule-based local rewrites. While their software is not open-source, when available the ``light'' optimization results reported in \cite{nrscm17} are given in \cref{tab:results}, as that algorithm most closely matches the scalability of ours. \Cref{alg:graysynth} typically results in similar $\cnot$ counts to their circuit optimizer, with \cref{alg:graysynth} reporting better $\cnot$ counts on some circuits (e.g., VE-Adder$\_3$, CSUM-MUX$\_9$) and worse on others (e.g., CSLA-MUX$\_3$, QCLA-Mod$\_7$). Additionally, it may be noted that in the case of the Galois field multipliers, the reductions Nam \etal achieve are from using base circuits with fewer $\cnot$ gates. As their optimizations rely largely on special-purpose synthesis and local rewrites, the techniques are complementary and so it may be possible to combine both to further reduce $\cnot$ counts. Further, most of their circuits are not verified in any way and hence may have errors, in comparison to ours which are almost all formally certified to be correct.

Nam \etal also report on the results of a ``heavy'' optimization algorithm which generally performs slightly better than \cref{alg:graysynth}, with the exception of the VBE-adder$\_3$ and Mod $5\_4$ benchmarks. This optimization however does not scale to the largest of our benchmark circuits, such as GF($2^{64}$)-Mult, as a result of the use of local rule-based rewrites to optimize $\cnot$-phase subcircuits. An interesting question is whether first performing a rough re-synthesis with \cref{alg:graysynth} would reduce run-times for the ``heavy'' local rewrites.

%%%%%%%%%%%%%%%%%%%%%%%%%%%%%%%%%%%%%%%%
\section{Conclusion}

In this paper we have shown that the problem of synthesizing a $\cnot$-minimal circuit over $\cnot$ and arbitrary angle $\rz$ gates, is at least as hard as synthesizing a minimal size $\cnot$ circuit computing a set of parities of the inputs. We have moreover shown that this problem is in fact NP-complete for the cases when all $\cnot$ gates have a fixed target, and when the inputs are encoded in some larger state space. As a result it would appear that the problem of optimizing $\cnot$ counts in quantum circuits is intractable. We further presented a heuristic algorithm to solve this problem which is inspired by Gray codes. The algorithm, when used in a sum-over-paths quantum circuit optimizer ($T$-par) to re-synthesize $\cnot$ and $\rz$ sub-circuits, reduces $\cnot$ counts by 23\% on average on our set of benchmarks. While the heuristic does not significantly outperform special-purpose synthesis and rule-based rewriting techniques \cite{nrscm17} on arithmetic benchmarks, our technique is novel and complementary to rewriting, opening new avenues for optimizing $\cnot$ gates by direct synthesis of parity networks.

While we suspect that the problem of synthesizing minimal parity networks is intractable, we leave the exact complexity as an open question. The two-dimensional nature of the problem makes it an unnatural problem for reductions, as efficient circuits make use of ``shortcuts'' by computing strategic parities into various bits. On the other hand, it may turn out to be possible to find an efficient algorithm for synthesizing minimal parity networks. Likewise, while we have presented an effective heuristic algorithm for synthesizing parity networks, we leave developing a heuristic algorithm synthesizing small pointed parity networks directly as a possibility for future work.

As a final point of consideration, while we have studied the problem of optimizing $\cnot$ gates in an unrestricted architecture, physical chip designs typically have limited connectivity and can only apply $\cnot$ gates between connected qubits. While arbitrary $\cnot$ gates can be implemented with a sequence of $\cnot$ gates having length at most proportional to the diameter of the connectivity graph, it nonetheless remains possible that a better circuit may be found by synthesizing one directly for a given topology. A natural and important direction for future research is then to find $\cnot$ optimization algorithms which take into account such connectivity constraints, a problem which likewise appears to be computationally intractable \cite{hnd17}. Going in the other direction, it may be possible to find more optimizations at the logical level by taking into account particular fault-tolerant models -- for instance, in lattice surgery-based circuits where multi-target $\cnot$ gates have cost strictly less than separate $\cnot$ gates \cite{hnd17}.

%%%%%%%%%%%%%%%%%%%%%%%%%%%%%%
%% Bibliography %%%%%%%%%%%%%%
%%%%%%%%%%%%%%%%%%%%%%%%%%%%%%

\bibliographystyle{habbrv}
\bibliography{main}

%%%%%%%%%%%%%%%%%%%%%%%%%%%%%%
%% Appendices %%%%%%%%%%%%%%
%%%%%%%%%%%%%%%%%%%%%%%%%%%%%%

\appendix

\section{Fourier expansions}\label{app:fourier}

The Fourier expansion of a pseudo-Boolean function used in this paper is not the one typically used in the analysis of Boolean functions \cite{od14}. In particular, the Boolean group is typically taken as the multiplicative group $\{-1, 1\}$ with the parity function $\parity{\y}:\{-1,1\}^n\rightarrow\{-1,1\}$ defined by $x_1^{y_1}x_2^{y_2}\cdots x_n^{y_n}$. The resulting expansion is then a \emph{multi-linear} polynomial in $\x$.

We can recover the standard Fourier expansion $f(\x) = \sum_{\y\in\F^n}\tilde{f}(\y)\widetilde{\parity{\y}}(\x)$ by defining $\widetilde{\parity{\y}}: \F^n\rightarrow \{-1, 1\}$ as
\[
  \widetilde{\parity{\y}}(\x) = (-1)^{x_1^{y_1}x_2^{y_2}\cdots x_n^{y_n}},
\]
in which case by observing that 
\[
  \fourier{f}(\y)\parity{\y}(\x) = \frac{1 + \fourier{f}(\y)}{2}\widetilde{\parity{\y}}(\x)
\]
  we see that 
\[
  \tilde{f}(\0) = \frac{2^n-1}{2}, \qquad \tilde{f}(\y) = \frac{1}{2}\fourier{f}(\y).
\]

\section{Proof of \cref{eq:incex}}\label{app:incex}

In this section we prove \cref{eq:incex}, namely that
$$ 2^{|\y|-1}\x^{\y} = \sum_{\y'\subseteq \y} (-1)^{|\y'|-1}\parity{\y'}(\x)$$
for any $\x,\y\in\F^n$.

\begin{lemma} For any $\x\in\F^n$, 
$$ 2^{n-1}x_1x_2\cdots x_n = \sum_{\y\in\F^n} (-1)^{|\y|-1}\parity{\y}(\x)$$
\end{lemma}
\begin{proof}
Clearly the formula is satisfied for $n=1$. Now consider $n=k+1$ for some $k$. Using the identity $x + y - (x\oplus y) = 2xy$ for any $x,y\in\F$ and basic arithmetic we observe that
\begin{align*}
  2^{k}x_1x_2\cdots x_{k+1}
  	&= 2^{k-1}x_{1}x_{2}\cdots(2x_{k}x_{k+1}) \\
  	&= 2^{k-1}x_{1}x_{2}\cdots(x_{k} + x_{k+1} - (x_{k} \oplus x_{k+1})) \\
  	&= 2^{k-1}x_{1}x_{2}\cdots x_{k} + 
  	      2^{k-1}x_{1}x_{2}\cdots x_{k+1} - 
  	      2^{k-1}x_{1}x_{2}\cdots(x_{k} \oplus x_{k+1})
\end{align*}

Next we define the length $k$ vectors $\x', \x'',\x'''\in\F^n$ as follows:
\[
 x'_i = x_i, \qquad
 x''_i = \begin{cases} x_{k+1} \text{ if } i=k \\ x_i \text{ otherwise} \end{cases}, \qquad
 x'''_i = \begin{cases} x_k \oplus x_{k+1} \text{ if } i=k \\ x_i \text{ otherwise} \end{cases}
\] 
By induction we see that
\begin{align*}
  2^{k}&x_1x_2\cdots x_{k+1} \\
    &= 2^{k-1}x'_{1}x'_{2}\cdots x'_k + 
  	  2^{k-1}x''_{1}x''_{2}\cdots x''_k - 
  	  2^{k-1}x'''_{1}x'''_{2}\cdots x'''_k \\
    &= \sum_{\y\in\F^k} (-1)^{|\y|-1}\left(\parity{\y}(\x') + \parity{\y}(\x'') - \parity{\y}(\x''')\right) \\
    &= \sum_{\substack{\y\in\F^k, \\ y_k = 0}} (-1)^{|\y|-1}\parity{\y}(\x') + \sum_{\substack{\y\in\F^k, \\ y_k = 1}} (-1)^{|\y|-1}\left(\parity{\y}(\x') + \parity{\y}(\x'') - \parity{\y}(\x''')\right) \\
    &= \sum_{\substack{\y\in\F^k, \\ y_k = 0}} (-1)^{|\y|-1}\left(\parity{\y}(\x')-\parity{\y}(\x')\oplus x_k - \parity{\y}(\x')\oplus x_{k+1} + \parity{\y}(\x') \oplus x_k\oplus x_{k+1}\right) \\
    &=\sum_{\y\in\F^{k+1}}(-1)^{|\y|-1}\parity{\y}(\x)
\end{align*}
\end{proof}

\begin{corollary}
For any $\x, \y\in\F^n$,
$$ 2^{|\y|-1}\x^{\y} = \sum_{\y'\subseteq \y} (-1)^{|\y'|-1}\parity{\y'}(\x)$$
\end{corollary}

\end{document}